%% file: ms.tex
\documentclass[12pt, english]{article}

\input{Preambles/header.tex}
\input{Preambles/notation.tex}

\title{Should We Adjust for the Test for Pre-trends in Difference-in-Difference Designs?}
\author{Jonathan Roth\thanks{Department of Economics, Harvard University, \href{mailto:jonathanroth@g.harvard.edu}{jonathanroth@g.harvard.edu}. I am grateful to Isaiah Andrews, Kirill Borusyak, Edward Glaeser, Nathan Hendren, Ariella Kahn-Lang, Max Kasy, Larry Katz, and Jann Spiess for helpful comments. This material is based upon work supported by the National Science Foundation Graduate Research Fellowship under Grant DGE1144152.}}

\begin{document}
\maketitle

\begin{abstract}
    The common practice in difference-in-difference (DiD) designs is to check for parallel trends prior to treatment assignment, yet typical estimation and inference does not account for the fact that this test has occurred. I analyze the properties of the traditional DiD estimator conditional on having passed (i.e. not rejected) the test for parallel pre-trends. When the DiD design is valid and the test for pre-trends confirms it, the typical DiD estimator is unbiased, but traditional standard errors are overly conservative. Additionally, there exists an alternative unbiased estimator that is more efficient than the traditional DiD estimator under parallel trends. However, when in population there is a non-zero pre-trend but we fail to reject the hypothesis of parallel pre-trends, the DiD estimator is generally biased relative to the population DiD coefficient. Moreover, if the trend is monotone, then under reasonable assumptions the bias from conditioning exacerbates the bias relative to the true treatment effect. I propose new estimation and inference procedures that account for the test for parallel trends, and compare their performance to that of the traditional estimator in a Monte Carlo simulation.   
\end{abstract}

\section{Introduction}

There is a standard playbook for difference-in-difference designs in applied economics. The researcher first identifies a treatment group and a candidate control group, where the treatment group received some treatment at period $t=1$ and the candidate control group did not receive the treatment but might plausibly have had parallel trends to the treatment group in the absence of the treatment. The researcher then tests whether the treatment and control group in fact had parallel trends prior to the treatment, and if the pre-trends appear to be (roughly) parallel, the researcher accepts the research design and interprets the difference-in-difference estimator as an estimate of the causal effect of treatment. 

While testing for pre-trends is nearly universal, typical estimation and inference for the treatment effects ignores the fact that this test has occurred. In this paper, I address the following question: do the usual properties we expect of our difference-in-difference estimator and its associated confidence interval hold conditional on having not rejected the test for parallel pre-trends? In particular, I discuss whether the point estimate is unbiased and most efficient among unbiased estimators, and whether the standard errors achieve their nominal coverage rate. I also propose new estimation and inference procedures that recover some of these properties in cases where they do not hold for the traditional difference-in-difference procedure.

I consider the performance of the classical difference-in-difference estimator (and alternatives) conditional on having passed the test for pre-trends in two different cases: first, when the parallel pre-trends assumption does in fact hold, and we have correctly confirmed it; and second, when the parallel pre-trends assumption is false but we have failed to reject it in our pre-test. 

When the parallel pre-trends assumption is satisfied in population, the traditional difference-in-difference estimator $\betahatpost$ remains unbiased conditional on the pre-test. However, in this case the traditional standard errors are overly conservative. Moreover, I show that there exists an unbiased estimator $\betatildepost$ that exploits the assumption of no pre-trends which is more efficient than the traditional estimator $\betahatpost$. 

When the parallel pre-trends assumption is false but we fail to reject it in our pre-test, $\betahatpost$ is generally biased relative to the population difference-in-difference coefficient, which itself will differ from the true causal effect of the treatment if the pre-trend would have continued into the post period. Moreover, if there is a monotone underlying trend, I show that the bias from conditioning exacerbates the existing bias relative to the true treatment effect under certain reasonable conditions on the covariance structure. To make this concrete, suppose that the true mean of $y$ is increasing by 1 unit per period relative to the control group in the pre-period, and this trend would have continued absent treatment. Unconditionally, we would expect the difference-in-difference estimator for the first period to be biased upwards by 1 relative to the true treatment effect. Conditional on not having rejected the parallel pre-trends assumption, this bias will be greater than 1. That is, the draws of the data in which we fail to reject the parallel trends assumption tend to produce particularly bad treatment effects estimates. Thus, the usual ``omitted variable bias'' formula does not hold once we condition on the pre-test. Unfortunately, the modified estimator $\betatildepost$ discussed above suffers from the same issue, and the bias can be even more severe.
 
I view these results as both ``good news'' and ``bad news'' for the applied researcher. In some cases, we have strong \textit{a priori} reason to believe that treatment assignment was random, and the test for pre-trends serves only as a sanity check. In these cases, it seems reasonable to use the modified estimator $\betatildepost$, which is more efficient when the parallel pre-trends assumption holds but has undesirable properties when it does not. In other cases, however, the researcher does not have such strong \textit{a priori} reasons to believe in parallel trends -- in many cases, the pre-test itself is one of the main justifications for the research design. The results here are problematic for this approach, because they show that this approach tends to fail to reject bad research designs precisely in the cases where the estimated treatment effects are farthest from the truth. This concern interacts with concerns about ``file drawer effects'' or ``publication bias,'' because even if the probability of accepting a given bad design is small, researchers may search over many bad designs until they find ones with apparently parallel trends. Researchers should therefore be wary of relying on the absence of an observable pre-trend as a justification in and of itself for the validity of a research design.

I also present two new estimators, along with corresponding confidence intervals, that adjust for the pre-test for parallel trends. The first procedure provides optimal median-unbiased estimates (and valid confidence intervals) for the population difference-in-difference coefficient, i.e. the probability limit of OLS without conditioning. It thus removes the additional bias imposed by conditioning on non-significant pre-trends when the parallel trends assumption fails, and hence allows us to make use of the usual intuition for how OLS difference-in-difference estimates are affected by non-parallel trends. Note, however, that this estimator fixes the bias from conditioning on the observed pre-trends but not the bias relative to the true treatment effect from accepting a bad research design -- if there is a pre-trend and it would have continued into the post period, the population difference-in-difference coefficient will differ from the true treatment effect. 

The second estimator modifies the first to obtain median-unbiased estimates of the OLS coefficient from a regression that allows the treatment group to be on a linear trend relative to the control group.\footnote{I show the estimator can also be modified to incorporate higher-order polynomial (e.g. quadratic) pre-trends.} If in population there is a linear pre-trend but we fail to reject the test for parallel trends, this estimator will remain unbiased for the true treatment effect under the assumption that the linear trend would have continued into the post period. If in population there is a pre-trend that is non-linear, the estimator will generally be biased relative to the treatment effect, but may serve as a good approximation to the extent that the extrapolation of the best linear fit to the pre-period serves as a close approximation to the counterfactual trend.

\textbf{Related Literature.} This paper contributes to a large literature in econometrics and statistics showing that, in a variety of contexts, problems can arise in both estimation and inference if researchers do not account for a pre-testing or model selection step. Previous research has investigated, for instance, the consequences of pre-testing for statistical significance to determine which covariates to include in a regression specification (\citet{giles_pre-test_1993} provide a review), automated variable selection via the LASSO \citep{lee_exact_2016}, and selection between many models using goodness-of-fit criteria such as AIC or BIC (see \citet{leeb_model_2005} and references therein). I note that similar issues arise in difference-in-differences, and show that under reasonable assumptions, when the parallel trends assumption is violated the bias induced by the pre-test exacerbates the bias of the difference-in-difference estimator relative to the true treatment effect.      

I provide two types of corrections to account for the fact that researchers test for observable pre-trends in difference-in-difference designs. The first exploits the assumption of parallel trends to obtain an unbiased treatment effects estimate $\betatildepost$ that is more efficient than the traditional estimator when the parallel trends assumption holds (both unconditionally, and conditional on observing a roughly parallel trend). The observation that the traditional difference-in-difference estimator is inefficient under parallel pre-trends mirrors a similar observation made by \citet{borusyak_revisiting_2016} in the context of event study designs. I show, however, that this type of estimator can perform quite poorly in the case where the parallel trends assumption does not hold but we fail to reject the hypothesis of parallel trends in our pre-test. The second type of correction provides optimal median-unbiased estimates and valid confidence intervals for the population difference-in-difference coefficient, or an adjustment thereof, even when the parallel pre-trends assumption is violated. These procedures build on results by \citet{lee_exact_2016} and \citet{andrews_identification_2017}, who respectively provide corrected inference procedures for the LASSO and for treatment effects estimates under publication bias, as well as on earlier work by \citet{pfanzagl_parametric_1994}, who developed optimal quantile-unbiased estimation results for exponential-family distributions.

By highlighting the need to account for the test for observably parallel pre-trends and extending the difference-in-difference framework to correct for this test, this paper also contributes to work on the econometrics of difference-in-difference designs in particular. Previous work has illustrated the need to adjust standard errors for serial correlation and other types of clustering in the data (see, e.g., \citet{moulton_illustration_1990, bertrand_how_2004, donald_inference_2007, petersen_estimating_2009}). Other research has, for instance, extended the traditional difference-in-difference framework to allow the treatment and control groups to have different benefits from treatment \citep{athey_identification_2006}, or to relax the typical identification assumption of parallel trends to allow for differences driven by observable characteristics \citep{abadie_semiparametric_2005}.

Finally, this paper relates to the literature on selective publication of scientific results (\citet{rothstein_publication_2005} and \citet{christensen_transparency_2016} provide reviews). Interest in this topic has been stoked recently by replication failures \citep{open2015estimating, Camereraaf0918}, as well as empirical evidence suggesting manipulation of p-values and substantial publication bias \citep{brodeur_star_2016, andrews_identification_2017}. The existing literature has primarily assumed (either implicitly or explicitly) that certain treatment effects estimates are inherently more publishable than others -- e.g. journals may prefer large or statistically significant treatment effects. My results suggest, however, that even if journal editors have no preferences over the estimated treatment effects \textit{per se}, the distribution of published treatment effects can still be affected if editors have preferences over the observed pre-trends coefficients. This paper therefore has important implications for the design of alternative publication processes to mitigate publication bias. For instance, \citet{nyhan2015increasing} proposes a results-blind peer review process, in which referees are provided information on the research design but not the resulting treatment effects. The results here suggest, however, that such a scheme may still create a distorted distribution of published estimates if referees are provided evidence on the pre-trends coefficients as evidence in favor of the research design.





\section{Intuition}


The intuition for why the distribution of estimated treatment effects is affected by the test for pre-trends stems from the fact that estimation error in the reference period ($t=0$) enters both the estimated pre-trends coefficients and the estimated treatment coefficients. To see this formally, let $\Delta \bar{Y}_t$ denote the difference in sample means between the treatment and control groups in period $t$, and let $\Delta \mu_t$ be the true difference in population means. The difference-in-difference estimate for the treatment effect in period 1 is $\betahatpost = \Delta \bar{Y}_1 - \Delta \bar{Y}_0 $, whereas the coefficient for the $k$th pre-period is $\betahat_{-k} = \Delta \bar{Y}_{-k} - \Delta \bar{Y}_0$ -- both of which contain a $\Delta \bar{Y}_0$ term. Although not necessary for most of our results, for simplicity consider the case where the estimation error $\Delta Y_t - \Delta \mu_t$ is independent across $t$, as would occur if the data come from repeated cross-sections. It is then immediate that $\betahatpost$ and $\betahat_{-k}$ will both tend to be above (below) their true expectation when $\Delta \bar{Y}_0 - \Delta \mu_0$ is less (greater) than zero. 

Conditioning on $\betahat_{-k}$ being insignificant thus has implications for the distribution of $\Delta \bar{Y}_0 - \Delta \mu_0$, and hence for $\betahat_{post}$. These dynamics are captured in Figure \ref{fig:distribution of measurement error}, which shows the simulated distribution of $\Delta \bar{Y}_0 - \Delta \mu_0$ and $\betahatpost - \betapost$ by whether or not we accept the pre-trends (meaning here that all of the $\betahat_{-k}$ are insignificant, in this case for $K=4$).\footnote{ Details on the simulation specification will be provided in Section \ref{sec:Simulations}.} In the top left panel, we see that when the true model is parallel, $\Delta \bar{Y}_{0} - \Delta \mu_0$ has lower variance once we condition on accepting the pre-trends, but continues to have mean 0. This is because under parallel trends, $\beta_{-k} = 0$, so $\betahat_{-k}$ tends to be insignificant for all $k$ when the measurement error $\Delta \bar{Y}_{0} - \Delta \mu_0$ is small in magnitude -- but the conditioning is symmetric around 0, so there is no bias. As a result, $\betahatpost - \betapost$ is centered around zero but has lower variance when we condition on accepting the pre-trends, as shown in the top right panel. By contrast, the bottom left panel shows that when there is truly an upward trend, then conditional on accepting the pre-trends, $\Delta \bar{Y}_{0} - \Delta \mu_0$ has mean below 0. This is because $\beta_{-k}<0$ for all $k$, so $\betahat_{-k}$ tends to be close to 0 when $\Delta \bar{Y}_{0} - \Delta \mu_0 < 0$. It follows that $\betahat_{post}$ tends to be greater than $\beta_{post}$ when there is a positive underlying trend but we fail to reject the pre-test for parallel trends, as shown in the bottom right panel.

\begin{figure}[!htp]
    \centering
    
    \subfloat{\includegraphics[width = 0.5\linewidth]{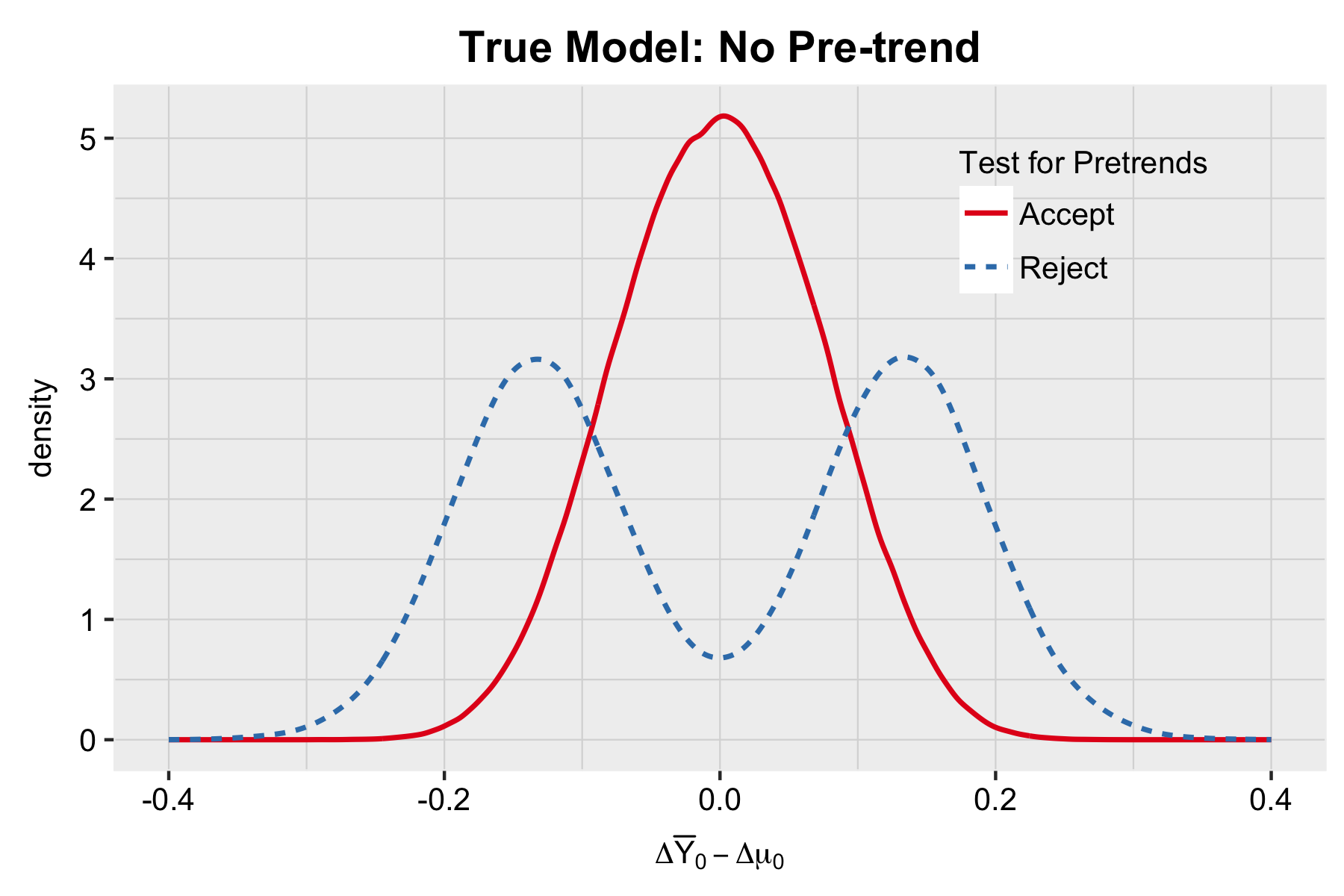}}
    \subfloat{\includegraphics[width = 0.5\linewidth]{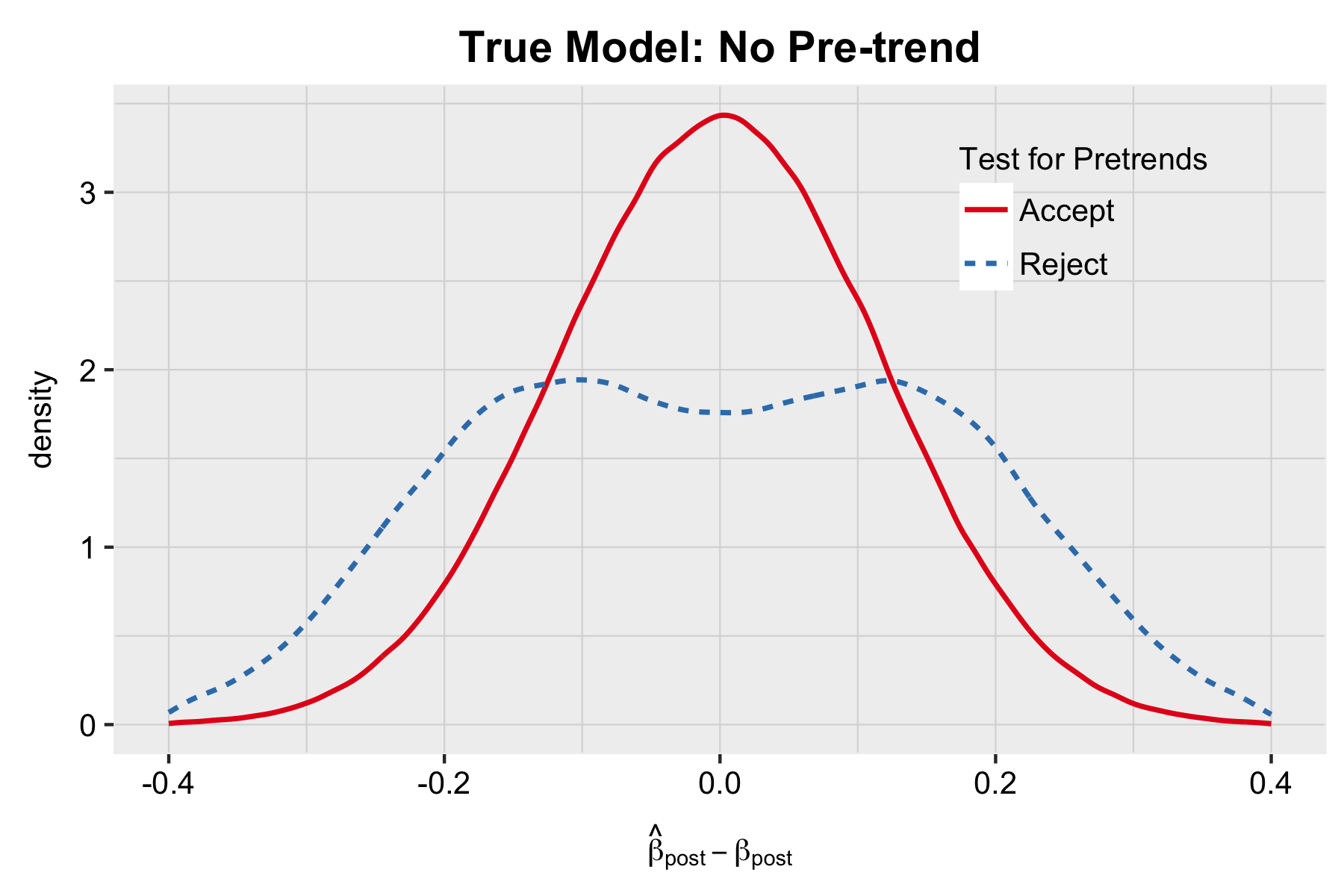}}\\
    \subfloat{\includegraphics[width = 0.5\linewidth]{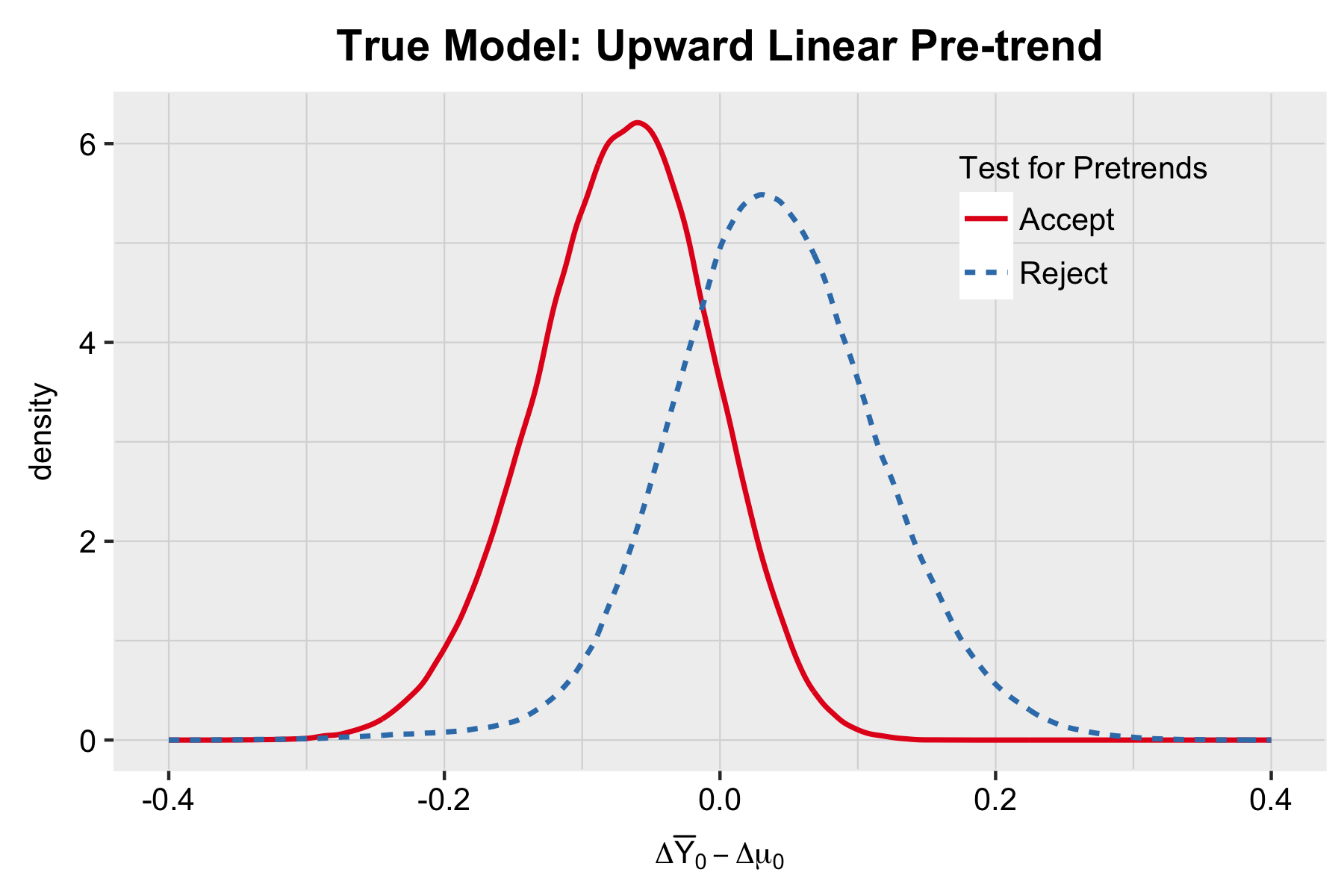}}
    \subfloat{\includegraphics[width = 0.5\linewidth]{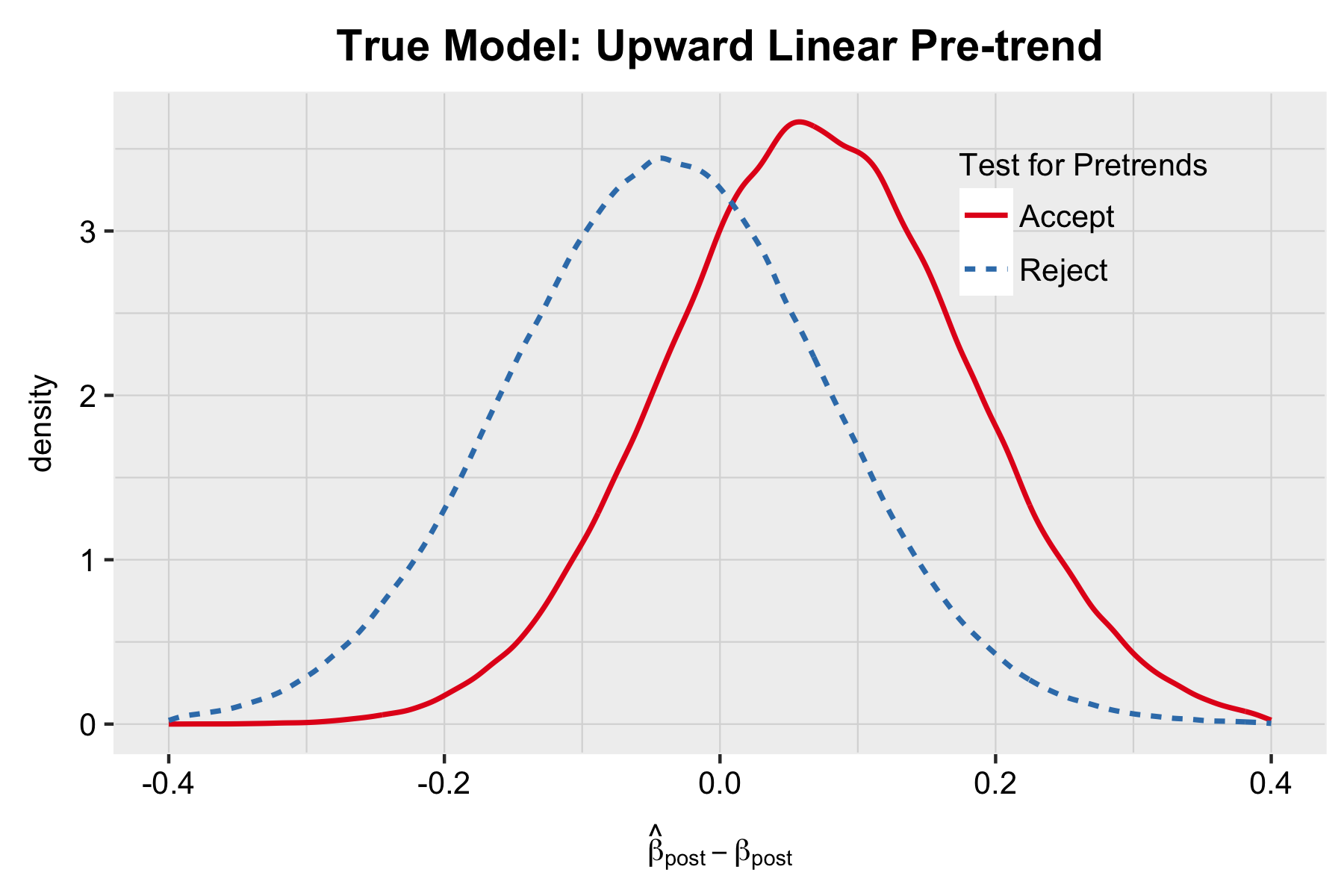}}
    \caption{Distribution of $\Delta \bar{Y}_0 - \Delta \mu_0$ and $\betahatpost - \betapost$ by whether or not we accept the test for parallel trends}
    \label{fig:distribution of measurement error}
\end{figure}

\section{Theoretical Results}

\subsection{Preliminary Set-Up}

We assume that the researcher observes data on an outcome $y$ for individuals in a treatment group and a control group, where the treatment group receives the treatment of interest beginning at period $t=1$. We suppose that the researcher observes $K+1$ periods prior to the treatment. For ease of exposition, we will assume that the researcher only observes a single period once the treatment occurs (i.e. $t=1$), although the results extend easily to estimating treatment effects in multiple subsequent periods. We suppose that the researcher estimates the following canonical regression: 

\begin{align} \label{eq: regression specification}
y_{it} = \alpha_t + \sum_{s \neq 0} \beta_s \times Treatment_i \times 1[t =s]  + \beta_{main} \, Treatment_i + \epsilon_{it}.
\end{align}

Let $\betahat_{pre} = (\betahat_{-1},\ldots,\betahat_{-K})$ denote the vector of estimated coefficients for the pre-periods, and $\betahatpost = \betahat_1$ denote the estimated coefficient for the post period. The researcher typically evaluates the pre-trends by examining whether $\betahatpre$ is close to zero by some metric -- e.g. whether all of its elements are statistically insignificant. If $\betahatpre$ passes this test, the researcher then evaluates $\betahatpost$ as an estimate of the causal effect of the treatment on $y$. We are interested in the distribution of $\betahatpost$ conditional on $\betahatpre$ having passed the test for parallel pre-trends.

Note that under very mild regularity conditions, the Central Limit Theorem will imply that (unconditionally):

$$\sqrt{N} \left( \twovec{\betahat_{post}}{\betahat_{pre}} - \twovec{\beta_{post}}{\beta_{pre}} \right) \to_{d} \normnot{ \twovec{0}{0} }{ \twobytwomat{ \tilde{\Sigma}_{11} }{ \tilde{\Sigma}_{12} }{ \tilde{\Sigma}_{21} }{ \tilde{\Sigma}_{22} } }.$$

For the purposes of this note, I assume that this normal approximation holds exactly, with known variance. That is, for fixed $N$,

$$\twovec{\betahat_{post}}{\betahat_{pre}} \sim \normnot{ \twovec{\beta_{post}}{\beta_{pre}} }{ \twobytwomat{ \Sigma_{11} }{ \Sigma_{12} }{ \Sigma_{21} }{ \Sigma_{22} } }$$

\noindent with $\Sigma = \twobytwomat{ \Sigma_{11} }{ \Sigma_{12} }{ \Sigma_{21} }{ \Sigma_{22} }$ known.\footnote{Note that $\Sigma$ here represents $\tilde{\Sigma} / N$ from the CLT above. That is, the square root of the diagonal of $\Sigma$ represents the standard errors.} 


\subsection{Results when the parallel pre-trends assumption is correct and we properly fail to reject it}

I first examine the performance of $\betahatpost$ conditional on having passed the test for pre-trends when the pre-period trends are in fact parallel in population.

\begin{asm}[Parallel Pre-trends]
    The parallel pre-trends assumption states that $\beta_{pre} = 0$.
\end{asm} 

We are interested in the distribution of $\betahat_{post}$ conditional on the researcher having accepted that the pre-period coefficients, $\betahat_{pre}$, are sufficiently close to zero. The standard used for determining whether pre-trends look ``okay'' in applied work is not entirely clear -- a commonly used criterion is that all the pre-period coefficients be (individually) statistically insignificant, although in other cases the researcher may follow other criteria (e.g. whether they can reject the joint test that all the pre-period coefficients are equal to 0). We will denote by $B_{NS} = \{ \hat{\beta}_{pre} : | \hat{\beta}_{pre,j} | / \sqrt{\Sigma_{jj}} \leq c_{\alpha} \text{ for all } j \}$ the set of values for $\betahat_{pre}$ such that no individual coefficient is statistically significant at the $\alpha$ level ($c_{\alpha} = 1.96$ for $\alpha = 0.05$), and we will denote by $B$ any arbitrary subset of the support of $\hat{\beta}_{pre}$. Some results will assume only that we've conditioned on $\hat{\beta}_{pre} \in B$ for some arbitrary $B$, whereas others will focus specifically on the case where we've conditioned on $\hat{\beta}_{pre} \in B_{NS}$.

\begin{prop}\label{prop: performance under parallel pre-trends}
    Suppose that the parallel pre-trends assumption holds. Let $\tilde{\beta}_{post} = \betahat_{post} - \Sigma_{12} \Sigma_{22}^{-1} \betahat_{pre} $. Then:
    
    \begin{enumerate}
        \item \label{subprop: betahat unbiased}
        $\expe{\betahat_{post} | \betahat_{pre} \in B_{NS}} = \beta_{post}.$
        
        \item \label{subprop: var betahat reduced}
        $ \var{\betahat_{post} | \betahat_{pre}  \in B_{NS} } \leq \Sigma_{11}  = \var{\betahat_{post}}. $
        
        \item \label{subprop: betatilde unbiased}
        $\expe{\tilde{\beta}_{post} | \betahat_{pre} \in B} = \beta_{post}$ for any $B$.
        
        \item \label{subprop: betatilde more efficient}
        $ \var{\betatilde_{post} | \betahat_{pre}  \in B } = \Sigma_{11} - \Sigma_{12} \Sigma_{22}^{-1} \Sigma_{21} \leq \var{\betahat_{post} | \betahat_{pre}  \in B } $ for any $B$, and the result holds with strict inequality if $\var{ \Sigma_{12} \Sigma_{22}^{-1} \betahatpre \,|\, \betahatpre \in B}  > 0 $.
    \end{enumerate}
    
\end{prop}

Result \ref{subprop: betahat unbiased} says that under the parallel pre-trends assumption, $\betahatpost$ remains unbiased if we condition on none of the pre-trend coefficients being significant. This result extends to conditioning on $\betahatpre \in B$ for any $B$ such that $\betahatpre$ has mean zero over the conditioning set. Result \ref{subprop: var betahat reduced} says that under the parallel pre-trends assumption, the variance of $\betahatpost$ conditional on not having significant pre-trend coefficients is smaller than the unconditional variance, which suggests that traditional standard errors may be conservative.\footnote{Note that this does not account for possible errors in the estimation of standard errors owing to the conditioning. The simulations below, which use estimates of $\Sigma$, suggest that this is likely second-order.} Result \ref{subprop: betatilde unbiased} says that under the parallel pre-trends assumption, the alternative estimator $\betatildepost$ is unbiased regardless of the conditioning set. Finally, Result \ref{subprop: betatilde more efficient} says that under the parallel pre-trends assumption, $\betatilde_{post}$ is more efficient than $\betahat_{post}$. Note that Results \ref{subprop: betatilde unbiased} and \ref{subprop: betatilde more efficient} apply for arbitrary conditoning sets, and thus, when the parallel pre-trends assumptions holds, $\betatildepost$ is both unbiased and more efficient than $\betahatpost$ even when we do not condition on the observed pre-trends.


\subsection{Results when the parallel pre-trends assumption is false but we fail to reject it}

I now consider the case where the parallel pre-trends assumption may not hold, i.e. when $\betapre$ is potentially not equal to 0.

\begin{prop}\label{prop: bias when parallel trends is false - general}
For any $\betapre$, with $\betapre$ potentially not equal to 0,    
    \begin{enumerate}
        \item \label{subprop: betahat biased when pretrends false} 
        $\expe{\betahatpost \,|\, \betahatpre \in B} = \betapost + \Sigma_{12} \Sigma_{22}^{-1} \left( \expe{\betahatpre \,|\, \betahatpre \in B} - \betapre \right) $ for any $B$.
        
        \item \label{subprop: betatilde biased when pretrends false}
        $\expe{\betatildepost \,|\, \betahatpre \in B} = \betapost - \Sigma_{12} \Sigma_{22}^{-1}  \betapre $ for any $B$.
        
        \item \label{subprop: betahat has lowe variance when condition on BNS - general}
        $ \var{\betahat_{post} | \betahat_{pre}  \in B_{NS} } \leq \Sigma_{11}  = \var{\betahat_{post}} $.

    \end{enumerate}
\end{prop}

Proposition \ref{prop: bias when parallel trends is false - general} gives a general formula for the expectation of $\betahatpost$ and $\betatildepost$, from which we see that if the parallel pre-trends assumption does not hold, both estimators will tend to be biased relative to the population difference-in-difference coefficient $\betapost$, which itself will differ from the true treatment effect if the pre-trend would have continued in the post period. Nonetheless, if we condition on none of the pre-treatment coefficients being significant, the true variance of $\betahatpost$ is smaller than the unconditional covariance. Thus, confidence intervals using the traditional standard errors may have coverage above or below the nominal rate depending on the relative size of the bias of the coefficient relative to that of the standard error -- a fact which we confirm in our simulations below.

I now turn to understanding the direction of the bias in the case where the pre-trends are monotone. To do this, I first introduce the following assumption on the covariance structure.

\begin{asm}\label{Assumption: covariance matrix structure}
$\Sigma$ has a common term $\sigma^2$ on the diagonal and a common term $\rho>0$ off of the diagonal, with $\sigma^2 > \rho$.
\end{asm} 

This is a natural assumption if the data used in the study come from repeated cross-sections (e.g. the ACS or Census), and the outcome of interest has constant variance across periods. To see why this is the case, again let $\Delta \bar{y_t}$ denote the difference in sample means between treatment and control in period $t$, and $\Delta \mu_t$ the corresponding difference in population. Note that for any $t$, $\betahat_t = \Delta \ybar_t - \Delta \ybar_0$, and its (unconditional) expectation is $\beta_t = \Delta \mu_t - \Delta \mu_0$. Let $\epsilon^\Delta_t = \Delta \ybar_t - \Delta \mu_t$. Then $\Cov \left( \betahat_j, \betahat_k \right) = \Cov \left( \epsilon^\Delta_j - \epsilon^\Delta_0 , \epsilon^\Delta_k - \epsilon^\Delta_0  \right) $. In repeated cross-sections, however, the measurement error across periods should be independent, so that for $j\neq k$, $\Cov \left( \betahat_j, \betahat_k \right) = \var{\epsilon^\Delta_0} \equiv \rho > 0$. Additionally, $\var{\betahat_k} = \var{\epsilon^\Delta_0} + \var{\epsilon^\Delta_k}$, which is constant in $k$ if $\var{\epsilon^\Delta_k}$ is the same across $k$. I note that for the result in Proposition \ref{prop: relative magnitudes under lienar pre-trends}, Assumption \ref{Assumption: covariance matrix structure} can be relaxed so that $\var{\epsilon^\Delta_k}$ need only be equal for all $k<0$, which allows for heterogenous treatment effects. I leave it to future work to determine whether the following result holds in a panel setting with serially correlated individual errors -- I conjecture that it should still hold so long as the serial correlation is not too large.

\begin{prop}[Sign of bias under upward pre-trend] \label{prop: relative magnitudes under lienar pre-trends}
Suppose that $\beta_{t} < 0$ for all $t<0$. If Assumption \ref{Assumption: covariance matrix structure} holds, then

\begin{enumerate}
    \item
     $\expe{\betahatpost \,|\, \betahatpre \in B_{NS}} > \betapost.$
    
    \item
     $\expe{\betatildepost \,|\, \betahatpre \in B_{NS}} > \betapost.$
    
\end{enumerate}

The analogous result holds replacing "$>$" with "$<$" and vice versa.    
\end{prop}

Proposition \ref{prop: relative magnitudes under lienar pre-trends} derives the direction of the bias from conditioning on not observing a significant pre-trend in the case where in population there is an upward pre-trend, in the sense that all of the population pre-period coefficients are less than zero. A leading example of this is when the pre-trend is monotone, although the condition in Proposition \ref{prop: relative magnitudes under lienar pre-trends} is weaker than this. It follows from the proposition that when there is truly an upward pre-trend, $\betahatpost$ and $\betatildepost$ are both biased upwards conditional on the pre-test, thus exacerbating the bias relative to the true treatment effect, assuming the monotone trend would have continued into the post period.

\subsection{Median-unbiased estimation and corrected confidence intervals under non-parallel pre-trends\label{subsec: TN confidence intervals}}

 In this section, I derive alternative estimators, along with associated confidence intervals, that correct for the bias induced by conditioning on having not observed a significant pre-trend when the parallel pre-trends assumption fails. 
 
 The first estimator is an optimal median-unbiased estimator for the population difference-in-difference coefficient $\betapost$, i.e. the population limit of OLS without conditioning. Although $\betapost$ is generally not equal to the treatment effect of interest when the parallel pre-trends assumption fails, the extent to which $\betapost$ differs from the treatment effect is given by the omitted variable bias formula, and so use of this estimator allows us to recover this intuition from OLS even when we condition on having observed non-significant pre-trends. 
 
 I also show that we can modify the first estimator to obtain a median-unbiased estimate of the population regression coefficient from a specification that allows the control group to have a polynomial trend in the pre-period. Consider, for instance, the regression specification:
 
 \begin{align} \label{eq: regression specification with trend}
y_{it} = \alpha_t +  \gamma_{pre} \times t \times Treatment_i + \sum_{s > 0} \gamma_s \times Treatment_i \times 1[t =s]  + \gamma_{main} \, Treatment_i + \epsilon_{it}.
\end{align}

\noindent If there is truly a linear  pre-trend, then the population OLS coefficient $\gammapost := \gamma_1$ will equal the treatment effect in the first post-period assuming the pre-trend would have continued absent treatment. If there is a non-linear pre-trend, then $\gammapost$ may serve as a close approximation to the treatment effect of interest to the extent that the best-fit linear pre-trend serves as a close approximation to the counterfactual trend. The second estimator I derive provides median-unbiased estimates and valid confidence intervals for $\gammapost$, conditional on having observed non-significant pre-trends. I also show that an analogous estimator can be derived for modifications of specification (\ref{eq: regression specification with trend}) that allow for higher-order (e.g. quadratic) polynomial pre-trends.  
 
In order to derive these estimators, it is first useful to note that we can write $\gammapost = \betapost - \eta_{\gamma}' \betapre$ for a suitably defined vector $\eta_{\gamma}$. Intuitively, in settings where each pre-period is given equal weight in the regression -- e.g. balanced panel or repeated cross-sections of equal size -- the difference between $\betapost$ in regression (\ref{eq: regression specification}) and $\gammapost$ in regression (\ref{eq: regression specification with trend}) will be equal to the predicted value at $t=1$ of the best-fit line through the pre-trends coefficients $\beta_{-K},\ldots, \beta_0$ (where $\beta_0 =0$). That is, letting $X$ be the matrix with 1s in the first column and $(0,-1,...,-K)$ in the second column, and $M_{-1}$ the selection matrix that selects all but the first column, $\eta_\gamma' = (1, 1)' (X'X)^{-1} X' M_{-1} $. Lemma \ref{lem: relationship between gamma and beta} in the Appendix formalizes the relationship between $\betapost$ and $\gammapost$, and shows that if we instead wanted to include a quadratic (or higher-order polynomial) trend in (\ref{eq: regression specification with trend}) in place of the linear trend, then we would merely need to add an additional column to $X$ with quadratic (or higher-order polynomial) time terms.

Having noted that $\betapost$ and $\gammapost$ are both linear combinations of $\beta = (\betapost, \betapre)$, we can make use of results by \citet{lee_exact_2016} and \citet{andrews_identification_2017} to derive the distribution of linear combinations of $\betahat$ conditional on $\betahatpre$ falling in $B_{NS}$. To do this, we first note that we can define a matrix $A^{NS}$ and vector $b^{NS}$ such that $\betahatpre \in B_{NS}$ iff $A^{NS} \betahat \leq b^{NS}$. In particular, it is easy to verify that this holds for $A^{NS} = \twobytwomat{0_{1 \times K }}{I_{K \times K}}{0_{1 \times K}}{-I_{K \times K}} $ and $b^{NS} = \twovec{c_\alpha \times \sqrt{diag(\Sigma)}}{c_\alpha \times \sqrt{diag(\Sigma)}}$. We now derive the distribution for linear combinations of $\betahat$ conditional on $A \betahat \leq b$ for arbitrary $A$ and $b$, noting that this nests conditioning on $\betahatpre \in B_{NS}$ as a special case.

\begin{prop}\label{prop: TN distribution}
Let $\betahat = (\betahatpost, \betahatpre)$ and $\eta  \neq 0$ be in $\reals^{K+1}$. Let $B = \{\beta | A \beta \leq b \} $ for some $A$ and $b$. Define $c = \Sigma \eta / (\eta' \Sigma \eta) $ and $Z = (I - c \eta') \betahat$. Then the distribution of $\eta' \betahat$ conditional on $\betahat  \in B$ and $Z$ is truncated normal with (untruncated) mean $\eta' \beta$ and (untruncated) variance $\eta' \Sigma \eta$, with truncation points $V^{-}$ and $V^{+}$ given by

\begin{align}
& V^{-}(z) = \underset{\{j : (Ac)_j < 0 \}}{max}  \frac{b_j - (Az)_j}{(Ac)_j} \\
& V^{+}(z) = \underset{\{j : (Ac)_j > 0 \}}{min}  \frac{b_j - (Az)_j}{(Ac)_j}.
\end{align}

\end{prop}

\begin{prop}
\label{prop: proper coverage of TN estimators}

Let $\eta  \neq 0$ be in $\reals^{K+1}$, and $B = \{\beta | A \beta \leq b \} $ for some $A$ and $b$. Assume that $\betahat \in B$ with positive probability, and that $\Sigma$ is full rank. Let $F_{\mu, \sigma^2}^{[l,u]}$ denote the CDF of the normal distribution with mean $\mu$ and variance $\sigma^2$ truncated to $[l,u]$. Define $\hat{b}_\alpha( \eta' \betahat, z)$ to be the value of $x$ that solves $F_{x, \eta' \Sigma \eta}^{[V^{-}(z),V^{+}(z)]}(\eta' \betahat) = \alpha$, for $V^{-}$ and $V^{+}$ as defined in Proposition \ref{prop: TN distribution}. Then for any $\alpha \in (0,1)$, 

$$P\left( \hat{b}_\alpha(\eta' \betahat, Z) \leq \eta' \beta \,|\, \betahat \in B \right) = \alpha.$$

Further, if the parameter space for $\beta$ is an open set, then $\hat{b}_\alpha$ is uniformly most concentrated in the class of level-$\alpha$ quantile-unbiased
estimators, in the sense that for any other level-$\alpha$ quantile unbiased estimator $\tilde{b}_\alpha$,
and any loss function $L(x, \eta' \beta)$ that attains its minimum at $x = \eta' \beta$ and is increasing as $x$ moves away from $\eta' \beta$,

$$\expe{L\left( \hat{b}_\alpha(\eta' \beta,Z) \, , \eta' \beta \right) \,|\, \betahat \in B} \leq \expe{L \left( \tilde{b}_\alpha(\eta' \beta,Z) \, , \eta' \beta \right) \,|\, \betahat \in B} .$$

Thus, conditional on $\betahat \in B$, $\hat{b}_{0.5}(\eta' \betahat, Z)$ is an (optimal) median-unbiased estimate of $\eta' \beta$, and the interval $[\hat{b}_{\alpha/2}(\eta' \betahat, Z) , \hat{b}_{1-\alpha/2}(\eta' \betahat, Z)]$ is a $1-\alpha$ confidence interval.
\end{prop}

\begin{cor}[Median Unbiased Estimator for $\betapost$] Let $\betahatpost^{TN} := \hat{b}_{0.5}
(\eta'\betahat,Z)$ be the estimator from Proposition \ref{prop: proper coverage of TN estimators} corresponding with $A = A^{NS}$, $b = b^{NS}$ and $\eta = e_1$, the first basis vector. Then, conditional on $\betahatpre \in B_{NS}$, $\betahatpost^{TN}$ is an optimal median unbiased estimator for $\betapost$. Likewise, $CI^{TN}_{\beta}:= [\hat{b}_{\alpha/2}
(\eta'\betahat,Z), \hat{b}_{1-\alpha/2}
(\eta'\betahat,Z)]$ is a $1-\alpha$ confidence interval for $\betapost$.
\end{cor}

\begin{cor}[Median Unbiased Estimator for $\gammapost$] Let $\gammahatpost^{TN} := \hat{b}_{0.5}
(\eta'\betahat ,Z)$ be the estimator from Proposition \ref{prop: proper coverage of TN estimators} corresponding with $A = A^{NS}$, $b = b^{NS}$ and $\eta = (1, \eta_\gamma)$. Then, conditional on $\betahatpre \in B_{NS}$, $\gammahatpost^{TN}$ is an optimal median unbiased estimator for $\gammapost$. Likewise, $CI^{TN}_{\gamma}:= [\hat{b}_{\alpha/2}
(\eta'\betahat,Z), \hat{b}_{1-\alpha/2}
(\eta'\betahat,Z)]$ is a $1-\alpha$ confidence interval for $\gammapost$.
\end{cor}



\section{Simulations\label{sec:Simulations}}

\subsection{Simulation DGP}
I conduct simulations in which we observe a control group and a treatment group for $K+1$ periods prior to treatment assignment and 1 period after treatment assignment. In the first specification, $y_{it} = \epsilon_{it}$ for $\epsilon \overset{iid}{\sim} N\left(0,\sigma^2\right)$, regardless of treatment status, so the parallel-trends assumption holds and there is no treatment effect. In the second specification, $y_{it} = \delta \times t \times Treatment_i + \epsilon_{it}$ for $\delta>0$, where again $\epsilon \overset{iid}{\sim} N\left(0,\sigma^2\right)$, so there is a linear upward pre-trend that continues into the post period. By construction, Assumption \ref{Assumption: covariance matrix structure} holds.

For each of 1 million simulated datasets, I estimate  regression (\ref{eq: regression specification}). Below, I report the distribution of $\betahatpost$ and its confidence intervals conditional on the $K$ pre-trends coefficients being statistically insignificant for $K = 0,\ldots,8$ (with $K=0$ denoting the unconditional case). I do the same for the modified estimator $\betatildepost$, as well as the median unbiased estimators and adjusted confidence intervals.

\subsection{Simulation results when the parallel trends assumption holds}

The left panel of Table \ref{tbl:betahat and betatilde under parallel trends} shows the distribution of $\betahatpost$ when the parallel trends assumption is true. The first row $(K=0)$ shows that without conditioning on the pre-trends, the average standard error matches the actual standard deviation of the estimator, and the associated confidence intervals achieve nominal coverage. Looking at the other rows of the table, however, we see that when we condition on insignificant pre-trends the average standard error is below the actual standard deviation of the estimator, and the null rejection probability (size) falls markedly below 0.05. This phenomenon becomes more acute as the number of pre-periods we condition on grows. In all cases, however, $\betahatpost$ is unbiased.

Turning to the right panel, we see that $\betatildepost$ is also unbiased, has lower variability than $\betahatpost$, and the estimated standard errors match the actual standard deviation of the estimator. In the $K=1$ case, these represent a 13\% reduction in estimated standard errors, and an 11\% reduction in the true variability of the estimator; in the $K=8$ case, these numbers rise to 26\% and 15\%. 

\begin{table}[!hbtp]
\caption{Performance of $\betahatpost$ and $\betatildepost$ when the parallel trends assumption is true and we condition on accepting pre-trends for $K$ periods\label{tbl:betahat and betatilde under parallel trends}}
\hspace{.5cm}
\begin{tabular}{@{\extracolsep{2pt}} L{.4cm} L{1.2cm}  L{1.2cm}  L{1.2cm} L{1.2cm} L{1.2cm}  L{1.2cm}  L{1.2cm} L{1.2cm}   }
     & \multicolumn{4}{c}{$\betahatpost$} & 
     \multicolumn{4}{c}{$\betatildepost$}\\ \cline{2-5} \cline{6-9} 
    K & 
    Bias & Mean S.E. & Actual S.D. & Size &
    Bias & Mean S.E. & Actual S.D. & Size \\
    \input{Tables/BetaPostandTilde_SummaryStats_WithoutPretrends.tex}

\end{tabular}
\end{table}

\subsection{Simulation results when there is actually an upward linear trend}

Table \ref{tbl:betahat and betatilde under linear pre-trend} shows the performance of $\betahatpost$ and $\betatildepost$ when the true model has an upward pre-trend. I set $\delta = 0.065$, so that the slope of the pre-trend is roughly half of the unconditional standard error for $\betahatpost$. 

Looking at the first row, we see that unconditionally $\betahatpost$ has a mean of $\betapost = 0.065$, i.e. precisely the slope of the pre-trend -- this is the usual ``omitted variable bias'' intuition that most people have in mind when evaluating the potential bias caused by non-parallel trends. However, as we move down the table we see that this bias becomes substantially larger once we condition on none of the pre-trends coefficients being statistically significant. Once we condition on 4 insignificant pre-period coefficients, we see that the conditional mean is more than twice that of the unconditional. 

As before, however, the standard error of $\betahatpost$ conditional on accepting the pre-trends is less than the mean estimated standard error. As a result, for $K$ small, when the bias is not too large, the traditional confidence intervals cover $\betapost$ more than 95\% of the time. However, for larger $K$ the bias dominates, and the CIs cover $\betapost$ less than 95\% of the time. In all specifications, the rejection probability at 0 exceeds 0.05 owing to the underlying trend, but the over-rejection problem becomes more extreme as we condition on more insignificant pre-period coefficients.

In the right panel, we see that while $\betatildepost$ outperformed $\betahatpost$ when the parallel trends assumption was true, its performance is even worse under the linear pre-trend. The bias relative to $\betapost$ is even larger than for $\betahatpost$, and the probability of rejecting $\betapost$ exceeds 0.05 in all cases.



\begin{table}[!hbt]
\caption{Performance of $\betahatpost$ and $\betatildepost$ when there is a linear upward pre-trend (slope = 0.065) but we do not reject parallel pre-trends for $K$ periods\label{tbl:betahat and betatilde under linear pre-trend}}
\hspace{.5cm}
\begin{tabular}{@{\extracolsep{2pt}} L{.4cm} L{1.1cm} L{1.1cm}  L{1.1cm}  L{1.1cm} L{1.1cm} L{1.1cm} L{1.1cm}  L{1.1cm}  L{1.1cm} L{1.1cm} L{1.1cm}   }
    & & \multicolumn{5}{c}{$\betahatpost$} & 
     \multicolumn{5}{c}{$\betatildepost$}\\ \cline{3-7} \cline{8-12}  
    K & Accept Pretrends &
    Mean $\betahatpost$ & Mean S.E. & Actual S.D. & Reject $\betapost$ & Reject 0 &
    Mean $\betatildepost$ & Mean S.E. & Actual S.D. & Reject $\betapost$ & Reject 0 \\[.2cm] \hline
   \input{Tables/BetaPostandTilde_SummaryStats_LinearPretrend.tex}

   \hline
\end{tabular}
\end{table}


\subsection{Performance of adjusted estimates and confidence intervals}

Tables \ref{tbl: median estimates for median-unbiased and betahatpost} and \ref{tbl: TN CIs coverage and width} show the performance of the median-unbiased estimates and truncated normal confidence intervals developed in Section \ref{subsec: TN confidence intervals} relative to the traditional point estimates and confidence intervals. The comparisons are shown for both the specification without any pre-trends and with the upward linear trend. In Table \ref{tbl: median estimates for median-unbiased and betahatpost}, we see that regardless of how many pre-trends coefficients we've conditioned on, $\betahatpostTN$ has a median very close to the true $\betapost$, i.e. 0 in the case without pre-trends and 0.065 in the case with a linear trend. Likewise, $\gammahatpostTN$ has a median close to 0 (i.e. $\gamma$) in all cases. Table \ref{tbl: TN CIs coverage and width} shows that $\CITNbeta$ and $\CITNgamma$ achieve (roughly) nominal rejection probabilities for $\betapost$ and $0$ respectively.

\begin{table}[!hbt]
\caption{Traditional and Median-Unbiased Point Estimates for $\betapost$ \label{tbl: median estimates for median-unbiased and betahatpost}}
\hspace{.5cm}
\begin{tabular}{@{\extracolsep{2pt}} L{.4cm} L{1.1cm} L{1.1cm}  L{1.1cm} L{1.1cm} L{1.1cm} L{1.1cm} }
 & \multicolumn{6}{c}{Median Estimate} \\\cline{2-7}
    & \multicolumn{3}{c}{No Pre-trend} & 
     \multicolumn{3}{c}{Linear Pre-trend}\\ \cline{2-4} \cline{5-7}\\[-.2cm]
   K  & $\betahatpost$ & $\betahatpost^{TN}$ & $\gammahatpost^{TN}$ & $\betahatpost$ & $\betahatpost^{TN}$ & $\gammahatpost^{TN}$ \\[.1cm] \hline 
 \input{Tables/TNversusTraditional_MedianBias_Table.tex}

   \hline
\end{tabular}
\end{table}

\begin{table}[!hbt]
\caption{Comparison of Traditional and Truncated Normal Confidence Intervals\label{tbl: TN CIs coverage and width}}
\hspace{.5cm}

\subfloat[Comparison of Traditional CIs versus $\CITNbeta$]{\begin{tabular}{@{\extracolsep{2pt}} L{.4cm} L{1.1cm} L{1.1cm}  L{1.1cm}  L{1.1cm}  L{1.1cm}  L{1.1cm}  L{1.1cm} L{1.1cm}    }
    & \multicolumn{4}{c}{No Pre-trend} & 
     \multicolumn{4}{c}{Linear Pre-trend}\\ \cline{2-5} \cline{6-9}\\
      & \multicolumn{2}{c}{Rejects $\betapost$} & \multicolumn{2}{c}{Median Width} &
      \multicolumn{2}{c}{Rejects $\betapost$} & \multicolumn{2}{c}{Median Width} \\[.2cm] \cline{2-3} \cline{4-5} \cline{6-7} \cline{8-9} 
      K&
      $\CITrad$ & $\CITNbeta$ & $\CITrad$ & $\CITNbeta$  & $\CITrad$ & $\CITNbeta$ & $\CITrad$ & $\CITNbeta$ \\[.2cm]
     \input{Tables/TN_versus_Standard_CIs.tex}

   \hline
\end{tabular}}\\
\subfloat[Comparison of Traditional CIs versus $\CITNgamma$]{\begin{tabular}{@{\extracolsep{2pt}} L{.4cm} L{1.1cm} L{1.1cm}  L{1.1cm}  L{1.1cm}  L{1.1cm}  L{1.1cm}  L{1.1cm} L{1.1cm}    }
    & \multicolumn{4}{c}{No Pre-trend} & 
     \multicolumn{4}{c}{Linear Pre-trend}\\ \cline{2-5} \cline{6-9}\\
      & \multicolumn{2}{c}{Rejects 0} & \multicolumn{2}{c}{Median Width} &
      \multicolumn{2}{c}{Rejects 0} & \multicolumn{2}{c}{Median Width} \\[.2cm] \cline{2-3} \cline{4-5} \cline{6-7} \cline{8-9} 
      K&
       $\CITrad$ & $\CITNgamma$ & $\CITrad$ & $\CITNgamma$  & $\CITrad$ & $\CITNgamma$ & $\CITrad$ & $\CITNgamma$ \\[.2cm]
     \input{Tables/TNAdjusted_versus_Standard_CIs.tex}

   \hline
\end{tabular}}

\end{table}

However, the improved properties of the $\CITNbeta$ and $\CITNgamma$ come at a power cost in many of the specifications, as evidenced by the median widths of the confidence intervals shown in Table \ref{tbl: TN CIs coverage and width}. 
In all specifications, the median widths for the traditional CIs are smaller than those for the $\CITNbeta$ intervals. The differences between the standard CIs and the $\CITNbeta$ intervals are fairly small when we condition on few periods, but increase as the number of conditioning periods increases. By contrast, the gap between the median widths of the traditional CIs and the $\CITNgamma$ intervals is large for small $K$, but decreases as $K$ grows. For $K$ above 6, the $\CITNgamma$ median widths are actually smaller than the traditional ones when the parallel pre-trends assumption holds. Interestingly, the median widths for both the $\CITNbeta$ and $\CITNgamma$ are much larger in the case with the linear pre-trend, so the simulations suggest that the penalty for using the adjusted estimates is smaller when the parallel pre-trends assumption actually holds. 




            
            
            
            

    
    






\section{Conclusion}


This paper illustrates that many of the properties we expect of the traditional OLS difference-in-difference estimator do not hold conditional on having passed the test for pre-trends that is standard in empirical work. I describe two approaches to address this issue. The first provides unbiased estimates and improved precision when the parallel trends assumption is correct (both conditional on the pre-test and unconditionally), but can have undesirable properties if the parallel trends assumption is false and we fail to reject it in our pre-test. I recommend using this estimator only in cases where the researcher has strong ex ante reasons to believe in the parallel trends assumption and the test of parallel trends serves only as a sanity check. The second approach, which I recommend in most cases, provides corrected point estimates and confidence intervals that condition on having not observed a significant pre-trend, thereby recovering many of the desirable properties of OLS even when the parallel trends assumption is false but we fail to reject it in our pre-test. I hope that these corrections will be useful for the applied researcher. I caution, however, that these corrections are not a panacea; a researcher who accepts a research design with a non-linear underlying trend will generally not recover an unbiased estimate of the treatment effect of interest even when applying these corrections. Therefore, in addition to applying these corrections, I also caution researchers not to exclusively rely on the test for pre-trends as validation for their research design, and to consider not just the significance of their tests for pre-trends but also the power of these tests to rule out meaningful violations of parallel trends.  

\FloatBarrier

\bibliography{Bibliography.bib}

\appendix
\section{Proofs}

\begin{lem}\label{lem:independence of betatilde and betapre}
Let $\tilde{\beta}_{post} = \betahat_{post} - \Sigma_{12} \Sigma_{22}^{-1} \betahat_{pre} $. Then $\betatildepost$ and $\betahatpre$ are independent.
\end{lem}

\begin{proof}
Note that by assumption, $\betahatpost$ and $\betahatpre$ are jointly normal. Since $\betatildepost$ is a linear combination of $\betahatpost$ and $\betahatpre$, it follows that $\betahatpre$ and $\betatildepost$ are jointly normal. It thus suffices to show that $\betahatpre$ and $\betatildepost$ are uncorrelated. We have:
\begin{align*}
\Cov \left( \betahat_{pre} \,,\, \tilde{\beta}_{post} \right) &= 
\expe{ 
\left(\betahat_{pre} - \beta_{pre} \right)
\left( (\betahat_{post} - \beta_{post} ) -  \Sigma_{12} \Sigma_{22}^{-1} (\betahat_{pre} - \beta_{pre}) \right)   } \\
&= \Sigma_{12} - \Sigma_{12} \Sigma_{22}^{-1} \Sigma_{22} \\
&= 0
\end{align*}

\end{proof}

We now prove Proposition \ref{prop: bias when parallel trends is false - general} for the case of generic $\betapre$, and then use these results to prove Proposition \ref{prop: performance under parallel pre-trends} for the special case where $\betapre = 0$.

\paragraph{Proof of Proposition \ref{prop: bias when parallel trends is false - general}}
\begin{enumerate}
    \item
    Note that by construction, $\betahatpost = \betatildepost + \Sigma_{12} \Sigma_{22}^{-1} \betahatpre$. It follows that

\begin{align*}
\expe{\betahatpost \,|\, \betahatpre \in B} &= \expe{\betatildepost \,|\, \betahatpre \in B} + \Sigma_{12} \Sigma_{22}^{-1} \expe{ \betahatpre \,|\, \betahatpre \in B} \\
&= \expe{\betatildepost } + \Sigma_{12} \Sigma_{22}^{-1} \expe{\betahatpre \,|\, \betahatpre \in B} \\
&= \expe{\betahatpost - \Sigma_{12} \Sigma_{22}^{-1} \betahatpre } + \Sigma_{12} \Sigma_{22}^{-1} \expe{\betahatpre \,|\, \betahatpre \in B} \\
&= \betapost - \Sigma_{12} \Sigma_{22}^{-1} \betapre  + \Sigma_{12} \Sigma_{22}^{-1} \expe{\betahatpre \,|\, \betahatpre \in B}\\
&= \betapost + \Sigma_{12} \Sigma_{22}^{-1} \left( \expe{\betahatpre \,|\, \betahatpre \in B} - \betapre \right)
\end{align*}

where the second line uses the independence of $\betatildepost$ and $\betahatpre$ from Lemma \ref{lem:independence of betatilde and betapre}, and the third and fourth use the definition of $\betatildepost$, $\betapost$, and $\betapre$.
    
    \item
    By Lemma \ref{lem:independence of betatilde and betapre}, $\betatildepost$ is independent of $\betahatpre$, so $\expe{\betatildepost \,|\, \betahatpre \in B} = \expe{\betatildepost}$ for any $B$. Further,
    
    $$\expe{\betatildepost} = \expe{\betahatpost - \Sigma_{12} \Sigma_{22}^{-1} \betahatpre } =  \betapost -  \Sigma_{12} \Sigma_{22}^{-1} \betapre .$$
    
    \item
    Note that since $\betahatpost = \betatildepost + \Sigma_{12}\Sigma_{22}^{-1} \betahatpre$, for any set $B$,
\begin{align}
\var{\betahatpost \,|\, \betahatpre \in B} &= \var{\betatildepost \,|\, \betahatpre \in B} + \var{\Sigma_{12} \Sigma_{22}^{-1} \betahatpre \,|\, \betahatpre \in B} \nonumber \\ & +  2 \Cov\left( \betatildepost \,,\,  \Sigma_{12} \Sigma_{22}^{-1} \betahatpre \,|\, \betahatpre \in B \right) \nonumber \\
&= \var{\betatildepost } + \var{ \Sigma_{12} \Sigma_{22}^{-1} \betahatpre \,|\, \betahatpre \in B} \label{eq: conditional variance formula for arbitrary B}
\end{align}

where we use the independence of $\betatildepost$ and $\betahatpre$ from Lemma \ref{lem:independence of betatilde and betapre} to obtain that $\var{\betatildepost \,|\, \betahatpre \in B} = \var{\betatildepost}$ and that the covariance term equals 0. It follows that
\begin{align*}
    \var{\betahat_{post} | \betahat_{pre}  \in B_{NS} } -  \var{\betahat_{post}} &= 
    \var{ \Sigma_{12} \Sigma_{22}^{-1} \betahatpre \,|\, \betahatpre \in B_{NS}} - \var{ \Sigma_{12} \Sigma_{22}^{-1} \betahatpre} \\ 
    &= (\Sigma_{12} \Sigma_{22}^{-1}) \left( \var{  \betahatpre \,|\, \betahatpre \in B_{NS}} - \var{\betahatpre} \right) (\Sigma_{12} \Sigma_{22}^{-1})'
\end{align*}

which is (weakly) negative since $\var{  \betahatpre \,|\, \betahatpre \in B_{NS}} - \var{\betahatpre}$ is negative semi-definite by Lemma \ref{lem: definiteness of jacobian of truncated expectation} below. $\Box$
    
    
\end{enumerate}

\begin{defn}[Symmetric Rectangular Truncation About 0]
We say that $B \subset \reals^K$ is a symmetric rectangular truncation around 0 if $B = \{ y \in \reals^K \,|\, - b_j \leq y_j \leq b_j, \text{ for } j = 1,\ldots,K \}$ for some non-negative constants $b_1,\ldots,b_K$.

\end{defn}
\begin{lem}\label{lem: mean of symmetric truncation equals 0}
Suppose $Y \sim \normnot{0}{\Sigma}$ is a K-dimensional multivariate normal, and $B$ is a symmetric rectangular truncation around 0. Then $\expe{Y \,|\, Y \in B} = 0$.
\end{lem}

\begin{proof}
Note that if $Y \sim \normnot{0}{\Sigma}$, then we also have $-Y \sim \normnot{0}{\Sigma}$. It follows that $\expe{Y \,|\, Y \in B} = \expe{-Y \,|\, (-Y) \in B}$. However, $- b_j \leq y_j \leq b_j$ iff $- b_j \leq -y_j \leq b_j$, so $(-Y) \in B$ iff $Y \in B$. It follows that 

\begin{align*}
    \expe{Y \,|\, Y \in B} &= \expe{-Y \,|\, (-Y) \in B} \\
    &= \expe{-Y \,|\, Y \in B} \\
    &= - \expe{Y \,|\, Y \in B}
\end{align*}

\noindent which implies that $\expe{Y \,|\, Y \in B} = 0$.

\end{proof}

\paragraph{Proof of Proposition \ref{prop: performance under parallel pre-trends}}

\begin{enumerate}
    \item
    From Proposition \ref{prop: bias when parallel trends is false - general}, Part \ref{subprop: betahat unbiased}, it suffices to show that $\expe{\betahatpre \,|\, \betahatpre \in B_{NS}} - \betapre$ equals 0. From the parallel pre-trends assumption, we have that $\betapre = 0$, so this is equivalent to $\expe{\betahatpre \,|\, \betahatpre \in B_{NS}} = 0$. Additionally, $\betahatpre$ is normally distributed with mean $\betapre = 0 $. Since $B_{NS}$ is a symmetric rectangular truncation around 0, it follows that $\expe{\betahatpre \,|\, \betahatpre \in B_{NS}} = 0$ by Lemma \ref{lem: mean of symmetric truncation equals 0}.
    
    \item
    This is a restatement of Proposition \ref{prop: bias when parallel trends is false - general}, Part \ref{subprop: betahat has lowe variance when condition on BNS - general} (for the special case where $\betapre = 0$).

    \item
    By Lemma \ref{lem:independence of betatilde and betapre}, $\betatildepost$ is independent of $\betahatpre$, so $\expe{\betatildepost \,|\, \betahatpre \in B} = \expe{\betatildepost}$ for any $B$. And 
    
    $$\expe{\betatildepost} = \expe{\betahatpost - \Sigma_{12} \Sigma_{22}^{-1} \betahatpre } =  \betapost -  \Sigma_{12} \Sigma_{22}^{-1} \betapre =  \betapost$$
    
    since $\betapre = 0 $ under parallel pre-trends.
    
    \item
    Note that by independence, $\var{\betatildepost \,|\, \betahatpre \in B} = \var{\betatildepost}$. Further,
    
    \begin{align*}
        \var{\betatildepost} &= \var{\betahatpost - \Sigma_{12} \Sigma_{22}^{-1} \betahatpre } \\
        &= \var{\betahatpost} + \left( \Sigma_{12} \Sigma_{22}^{-1} \right) \var{\betahatpre} \left( \Sigma_{12} \Sigma_{22}^{-1} \right)' - 2 \Sigma_{12} \Sigma_{22}^{-1} \Cov\left(\betahatpost , \betahatpre \right)\\ &=
        \Sigma_{11} + \Sigma_{12} \Sigma_{22}^{-1} \Sigma_{21} - 2 \Sigma_{12} \Sigma_{22}^{-1} \Sigma_{21} \\ &=
        \Sigma_{11} - \Sigma_{12} \Sigma_{22}^{-1} \Sigma_{21}
    \end{align*}
    
    To prove the desired inequality, note that from (\ref{eq: conditional variance formula for arbitrary B}), 
    \begin{align*}
        \var{\betahatpost \,|\, \betahatpre \in B} - \var{\betatildepost} =
        \var{\Sigma_{12} \Sigma_{22}^{-1} \betahatpre \,|\, \betahatpre \in B} 
     \end{align*}
\noindent which is weakly positive. $\Box$ 
\end{enumerate}

We now prove a series of Lemmas leading up to the proof of Proposition \ref{prop: relative magnitudes under lienar pre-trends}.

\begin{lem}\label{lem: definiteness of jacobian of truncated expectation}
Suppose $Y$ is a k-dimensional multivariate normal, $Y \sim \normnot{\mu}{\Sigma}$, and let $B \subset \reals^k$ be a convex set such that $\prob{Y \in B} > 0$. Letting $D_\mu$ denote the Jacobian operator with respect to $\mu$, we have

\begin{enumerate}
        \item
    $D_\mu \expe{Y \,|\, Y \in B, \mu } = \var{Y \,|\, Y \in B , \mu} \Sigma^{-1}$.
    
    \item
    $\var{Y \,|\, Y \in B} - \Sigma$ is negative semi-definite.
\end{enumerate}
\end{lem}

\textit{Proof.}\footnote{I am grateful to Alecos Papadopolous, whose \href{https://math.stackexchange.com/questions/445164/is-the-mean-of-the-truncated-normal-distribution-monotone-in-mu}{answer} on StackOverflow to a related question inspired this proof.}

Define the function $H: \reals^{k} \rightarrow \reals$ by
$$H(\mu) = \int_B \phi_\Sigma(y - \mu) dy $$

\noindent for $\phi_\Sigma(x) = det(2\pi \Sigma)^{-\frac{1}{2}} exp(-\frac{1}{2} x' \Sigma^{-1} x) $ the PDF of the $\normnot{0}{\Sigma}$ distribution. We now argue that $H$ is log-concave in $\mu$. Note that we can write $H(\mu) = \int_{\reals^k} g_1(y, \mu) g_2(y,\mu) dy$ for $g_1(y,\mu) = \phi_\Sigma(y- \mu)$ and $g_2(y,\mu) = 1\left[y \in B\right]$. The normal PDF is log-concave, and $g_1$ is the composition of the normal PDF with a linear function, and hence log-concave as well. Likewise, $g_2$ is log-concave since $B$ is a convex set. The product of log-concave functions is log-concave, and the marginalization of a log-concave function with respect to one of its arguments is log-concave by Prekopa's theorem (see, e.g. Theorem 3.3 in \citet{saumard_log-concavity_2014}), from which it follows that $H$ is log-concave in $\mu$.

Now, applying Leibniz's rule, we have that the $1 \times k$ gradient of $\log H$ with respect to $\mu$ is equal to

\begin{align*}
    D_\mu \log H &= \dfrac{\int_B D_\mu \phi_\Sigma(y-\mu) dy}{\int_B \phi_\Sigma(y-\mu) dy}\\
    &= \dfrac{\int_B \phi_\Sigma(y-\mu) (y - \mu)' \Sigma^{-1} dy}{\int_B \phi_\Sigma(y-\mu) dy}\\
    &= (\expe{Y \,|\, Y \in B} - \mu)' \Sigma^{-1}.
\end{align*}

\noindent where the second line uses the fact that $D_\mu \phi_{\Sigma}(y - \mu) = \phi_{\Sigma}(y - \mu) \cdot (y - \mu)' \Sigma^{-1}$ and the third uses the definition of the conditional expectation. It follows that

\begin{align*}
   \expe{Y \,|\, Y \in B, \mu} = \mu + \Sigma (D_\mu \log H)'. 
\end{align*}

Differentiating again with respect to $\mu$, we have that the $k \times k$ Jacobian of $\expe{Y \,|\, Y \in B, \mu}$ with respect to $\mu$ is given by

\begin{align}
  D_\mu \expe{Y \,|\, Y \in B, \mu} = I + \Sigma D_\mu (D_\mu \log H)'. \label{eq: jacobian of expectation in terms of hessian} 
\end{align}

\noindent Since $H$ is log-concave, $D_\mu (D_\mu \log H)'$ is the Hessian of a concave function, and thus is negative semi-definite. Next, note that by definition,

\begin{align*}
    \expe{Y \,|\, Y \in B, \mu} = \dfrac{\int_B y \, \phi_\Sigma (y - \mu) \, dy }{ \int_B \phi_\Sigma (y - \mu) \, dy }.
\end{align*}

\noindent Thus, by the product rule,

\begin{align}
    D_\mu \expe{Y \,|\, Y \in B, \mu} =& \dfrac{\int_B y \, D_\mu \phi_\Sigma (y - \mu) \, dy }{ \int_B \phi_\Sigma (y - \mu) \, dy } + \nonumber \\
    & \left[ \int_B y \, \phi_\Sigma (y - \mu) \, dy \right] \cdot
    D_\mu \left[ \int_B \phi_\Sigma (y - \mu) \, dy  \right]^{-1}. \label{eq: derivative of expectation, long form}
\end{align}

\noindent Recall that

$$D_\mu \phi_{\Sigma}(y - \mu) = \phi_{\Sigma}(y - \mu) \cdot (y - \mu)' \Sigma^{-1}. $$

\noindent The first term in (\ref{eq: derivative of expectation, long form}) thus becomes

\begin{align*}
     \dfrac{\int_B y (y - \mu)' \phi_\Sigma (y - \mu) \, dy }{ \int_B \phi_\Sigma (y - \mu) \, dy } \Sigma^{-1} =\\
    \left( \expe{Y Y' \,|\, Y \in B, \mu} - \expe{Y \,|\, Y \in B, \mu} \mu' \right) \Sigma^{-1}.
\end{align*}

Applying the chain-rule, the second term in (\ref{eq: derivative of expectation, long form}) becomes

    \begin{align*}
-\dfrac{\int_B y \, \phi_\Sigma (y - \mu) \, dy \cdot \int_B (y - \mu)' \, \phi_\Sigma (y - \mu) \, dy }{\left[ \int_B \phi_\Sigma (y - \mu) \, dy  \right]^{2}} \Sigma^{-1} =\\
\left( - \expe{Y \,|\, Y \in B, \mu } \expe{Y \,|\, Y \in B, \mu }' + \expe{Y \,|\, Y \in B, \mu } \mu' \right) \Sigma^{-1}.
\end{align*}

Substituting back into (\ref{eq: derivative of expectation, long form}), we have

\begin{align}
D_\mu \expe{Y \,|\, Y \in B, \mu} &= \left( \expe{Y Y' \,|\, Y \in B, \mu } - \expe{Y \,|\, Y \in B, \mu} \expe{Y \,|\, Y \in B, \mu}' \right) \Sigma^{-1} \nonumber \\
&= \var{Y \,|\, Y \in B, \mu} \Sigma^{-1} \label{eq: jacobian of expectation equals truncated var}
\end{align}

\noindent which establishes the first result. Additionally, combining (\ref{eq: jacobian of expectation in terms of hessian}) and (\ref{eq: jacobian of expectation equals truncated var}), we have that

\begin{align}
\var{Y \,|\, Y \in B, \mu} \Sigma^{-1} = I + \Sigma D_\mu (D_\mu \log H)' \label{eq: variance of Y times Sigmainv in terms of hessian}
\end{align}

\noindent which implies that

\begin{align}
\var{Y \,|\, Y \in B, \mu} - \Sigma = \Sigma \,  D_\mu (D_\mu \log H)' \,  \Sigma. \label{eq: difference of variances}
\end{align}

\noindent Thus, for any vector $x \in \reals^k$, 

\begin{align*}
    x' \left( \var{Y \,|\, Y \in B, \mu} - \Sigma \right) x &= 
    x' \left( \Sigma \, D_\mu (D_\mu \log H)' \, \Sigma \right) x \\
    &= (\Sigma x)' \, \left(D_\mu (D_\mu \log H)' \right) \, (\Sigma x)
    \\&\leq 0 
\end{align*}

\noindent where the inequality follows from the fact that $D_\mu (D_\mu \log H)'$ is negative semi-definite. Since $\var{Y \,|\, Y \in B, \mu} - \Sigma$ is symmetric, it follows that it is negative semi-definite, as we desired to show. $\Box$

\begin{lem}\label{lem: Covariance structure assumption implies iota an eigenvector}

Suppose that $\Sigma$ satisfies Assumption \ref{Assumption: covariance matrix structure}. Then for $\iota$ the vector of ones and some $c_1>0$, $\iota' \Sigma^{-1} = c_1 \iota'$. Additionally, $\Sigma_{12} \Sigma_{22}^{-1} = c_2 \iota'$, for a constant $c_2 > 0$. 

\end{lem}

\begin{proof}

Note that we can write $\Sigma = \Lambda + \rho \iota \iota'$, where $\Lambda = (\sigma^2 - \rho) I$. It follows from the Sherman-Morrison formula that:

\begin{align*}
 \Sigma^{-1} &= \Lambda^{-1} - \frac{\rho^2 \Lambda^{-1} \iota \iota' \Lambda^{-1} }{1 + \rho^2 \iota' \Lambda^{-1} \iota}    \\
 &= (\sigma^2 - \rho)^{-1} I - \frac{\rho^2 (\sigma^2 - \rho)^{-2} \iota \iota' }{1 + \rho^2 (\sigma^2 - \rho)^{-1} \iota' \iota}.
\end{align*}

\noindent Thus:
\begin{align*}
& \iota' \Sigma^{-1} = \\
&  \iota' \left( (\sigma^2 - \rho)^{-1} I - \frac{\rho^2 (\sigma^2 - \rho)^{-2} \iota \iota' }{1 + \rho^2 (\sigma^2 - \rho)^{-1} \iota' \iota} \right) =\\
& (\sigma^2 - \rho)^{-1} \left(1 - \frac{\rho^2 (\sigma^2 - \rho)^{-1} \iota' \iota }{1 + \rho^2 (\sigma^2 - \rho)^{-1} \iota' \iota} \right) \iota' = \\
&\underbrace{ (\sigma^2 - \rho)^{-1} \left( \frac{1}{1 + \rho^2 (\sigma^2 - \rho)^{-1} \iota' \iota} \right) }_{:=c_1} \, \iota'.   
\end{align*}

\noindent Since $\sigma^2 - \rho >0$, all of the terms in $c_1$ are positive, and thus $c_1>0$, as needed. Additionally, note that if $\Sigma$ satisfies Assumption \ref{Assumption: covariance matrix structure}, then $\Sigma_{22}$ also satisfies Assumption \ref{Assumption: covariance matrix structure} and $\Sigma_{12} = \rho \iota'$. It follows that $\Sigma_{12} \Sigma_{22}^{-1} = \rho c_1 \iota' = c_2 \iota'$ for $c_2 = \rho c_1 > 0$.

\end{proof}

\begin{lem}\label{lem: mean of normal shifted upwards positive}
Suppose $Y \sim N\left(0 , \Sigma \right)$ for $\Sigma$ satisfying Assumption \ref{Assumption: covariance matrix structure}. Let $B = \{ y \in \mathbb{R}^K \, | \, a_j \leq y \leq b_j \text{ for all } j \} $, where $- b_j < a_j < b_j$ for all $j$. Then for $\iota$ the vector of ones,  $\expe{\iota' Y \,|\, Y \in B} = \expe{Y_1 + \ldots + Y_k \,|\, Y \in B}$ is elementwise greater than 0. 
\end{lem}

\begin{proof}

For any $x \in \mathbb{R}^K$ such that $x_j \leq b_j$ for all $j$, define $B^X(x) = \{ y \in \mathbb{R}^K \, | \, x_j \leq y \leq b_j \text{ for all } j \}  $. Let $b = (b_1, \ldots, b_K)$. Note that $B^X(-b)$ is a symmetric rectangular truncation around 0, so from Lemma \ref{lem: mean of symmetric truncation equals 0}, we have that $\expe{Y \,|\, Y \in B^X(-b) } = 0$. Now, define

\begin{align*}
g(x) = \expe{\iota' Y \,|\, Y \in B^X(x)}.   
\end{align*}

\noindent From the argument above, we have that $g(-b) = 0$, and we wish to show that $g(a) > 0$. Note that by the mean-value theorem, for some $t \in (0,1)$,
\begin{align*}
    g(a) &= g(-b) + (a - (-b)) \; \nabla g\left( t a + (1-t) (-b) \right) \\
    &= (a + b) \nabla g\left( t a + (1-t) (-b) \right) \\
    &=: (a+b) \nabla g(x^t).
\end{align*}

By assumption, $(a+b)$ is elementwise greater than 0. It thus suffices to show that all elements of $\nabla g\left( \xt \right)$ are positive. WLOG, we show that  $\dfrac{ \partial g(\xt) }{ \partial x_K } > 0$. 

Using the definition of the conditional expectation and Leibniz's rule, we have 
{\small 
\begin{align}
    &\dfrac{ \partial g(\xt) }{ \partial x_K } = \nonumber \\& \dfrac{ \partial }{ \partial \xt_K } \left[  \left( \int_{\xt_1}^{b_1} \cdots \int_{\xt_K}^{b_K} (y_1 + \ldots + y_K) \; \phi_\Sigma(y) \; dy_1 \ldots dy_K  \right) \left( \int_{\xt_1}^{b_1} \cdots \int_{\xt_K}^{b_K}  \; \phi_\Sigma(y) \; dy_1 \ldots dy_K  \right)^{-1}  \right] = \nonumber \\ 
    &   \left( \int_{\xt_1}^{b_1} \cdots \int_{\xt_K}^{b_K} (y_1 + \ldots + y_K) \; \phi_\Sigma(y) \; dy_1 \ldots dy_K  \times  \int_{\xt_1}^{b_1} \cdots \int_{ \xt_{K-1} }^{b_{K-1}}  \; \phi_\Sigma\left( \twovec{y_{-K}}{\xt_K} \right) \; dy_1 \ldots dy_{K-1}  \right. \nonumber \\
    & \left. -    \int_{\xt_1}^{b_1} \cdots \int_{ \xt_{K-1} }^{b_{K-1}}  (y_1 + \ldots + y_{K-1} + \xt_K) \; \phi_\Sigma\left(\twovec{y_{-K}}{\xt_K}\right) \; dy_1 \ldots dy_{K-1} \;\times \; \int_{\xt_1}^{b_1} \cdots \int_{\xt_K}^{b_K}  \; \phi_\Sigma(y) \; dy_1 \ldots dy_K          \right)  \nonumber\\
    & \times \left( \int_{\xt_1}^{b_1} \cdots \int_{\xt_K}^{b_K}  \; \phi_\Sigma(y) \; dy_1 \ldots dy_K  \right)^{-2} \label{eq:dg dxK}
\end{align}
}

\noindent where $\phi_\Sigma(y)$ denotes the PDF of a multivariate normal with mean 0 and variance $\Sigma$. It follows from (\ref{eq:dg dxK}) that $\dfrac{ \partial g(\xt) }{ \partial x_K } > 0$ if and only if

\begin{align*}
    &\dfrac{\int_{\xt_1}^{b_1} \cdots \int_{\xt_k}^{b_K} (y_1 + \ldots + y_K) \; \phi_\Sigma(y) \; dy_1 \ldots dy_K}{\int_{\xt_1}^{b_1} \cdots \int_{\xt_k}^{b_K}  \; \phi_\Sigma(y) \; dy_1 \ldots dy_K} > \\
    &\dfrac{\int_{\xt_1}^{b_1} \cdots \int_{ \xt_{K-1} }^{b_{K-1}}  (y_1 + \dots + y_{K-1} + \xt_K) \; \phi_\Sigma\left(\twovec{y_{-K}}{\xt_K}\right) \; dy_1 \ldots dy_{K-1}}{\int_{\xt_1}^{b_1} \cdots \int_{ \xt_{K-1} }^{b_{K-1}} \; \phi_\Sigma\left(\twovec{y_{-K}}{\xt_K}\right) \; dy_1 \ldots dy_{K-1}}
\end{align*}

\noindent or equivalently,

\begin{align*}
\expe{Y_1 + \ldots + Y_K \,|\, \xt_j \leq Y_j \leq b_j, \forall j} > \expe{Y_1 + \ldots + Y_K \,|\, \xt_j \leq Y_j \leq b_j, \text{ for } j < K, \, Y_K = \xt_K}.    
\end{align*}

\noindent It is clear that $\expe{Y_K \,|\, \xt_j \leq Y_j \leq b_j, \forall j} > \xt_K$, since $\xt_K < b_K$ and the $K$th marginal density of the rectangularly-truncated normal distribution is positive for all values in $[\xt_K, b_K]$ (see \citet{cartinhour_one-dimensional_1990}). It thus suffices to show that

\begin{align}
\expe{Y_1 + \ldots + Y_{K-1} \,|\, \xt_j \leq Y_j \leq b_j, \forall j} \geq \expe{Y_1 + \ldots + Y_{K-1} \,|\,  \xt_j \leq Y_j \leq b_j, \text{ for } j < K, \, Y_K = \xt_K}. \label{eqn: inequality with truncated expectation of y1 to yk}
\end{align}

To see why (\ref{eqn: inequality with truncated expectation of y1 to yk}) holds, let $\tilde{Y}_{-K} = Y_{-K} - \Sigma_{-K,K} \Sigma_{K,K}^{-1} Y_K$, where a ``$-K$'' subscript denotes all of the indices except for $K$. By an argument analogous to that in the Proof of Lemma \ref{lem:independence of betatilde and betapre} for $\betatildepost$, one can easily verify that $\tilde{Y}_{-K} \sim \normnot{0}{\tilde{\Sigma}}$ for some $\tilde{\Sigma}$, and $\tilde{Y}_{-K}$ is independent of $Y_K$. Additionally, by an argument analogous to that in Proposition \ref{prop: performance under parallel pre-trends}, Part \ref{subprop: betatilde more efficient}, one can show that $\tilde{\Sigma} = \Sigma_{-K,-K} - \Sigma_{-K,K} \Sigma_{K,K}^{-1} \Sigma_{K,-K}$. We now argue that $\tilde{\Sigma}$ satisfies Assumption \ref{Assumption: covariance matrix structure}. Since $\Sigma$ satisfies Assumption \ref{Assumption: covariance matrix structure}, so too does $\Sigma_{-K,-K}$. Additionally, under Assumption \ref{Assumption: covariance matrix structure}, $\Sigma_{-K,K} = \rho \iota$ and $\Sigma_{K,K}^{-1} = \frac{1}{\sigma^2}$, so $\Sigma_{-K,K} \Sigma_{K,K}^{-1} \Sigma_{K,-K}$ equals $\rho^2/ \sigma^2$ times $\iota \iota'$, the matrix of ones. The diagonal terms of $\tilde{\Sigma}$ are thus equal to $\tilde{\sigma}^2 = \sigma^2 - \rho^2/\sigma^2$, and the off-diagonal terms are equal to $\tilde{\rho} = \rho - \rho^2/\sigma^2$, or equivalently $\tilde{\rho} = \rho (1 - \rho/ \sigma^2)$. Since by Assumption \ref{Assumption: covariance matrix structure}, $0 < \rho < \sigma^2$, it is clear that $\tilde{\sigma}^2 > \tilde{\rho}$. Additionally, $0 < \rho < \sigma^2$ implies that $1 - \rho/\sigma^2 >0$, and hence $\tilde{\rho} > 0$, which completes the proof that $\tilde{\Sigma}$ satisfies Assumption \ref{Assumption: covariance matrix structure}. 

Now, by construction, $Y_{-K} = \tilde{Y}_{-K} + \Sigma_{-K,K} \Sigma_{K,K}^{-1} Y_K$. By Lemma \ref{lem: Covariance structure assumption implies iota an eigenvector}, under Assumption \ref{Assumption: covariance matrix structure},  $\Sigma_{-K,K} \Sigma_{K,K}^{-1} = c \iota$, for some $c>0$, so $Y_{-K} = \tilde{Y}_{-K} + c \iota Y_{K}$. From independence of $\tilde{Y}_{-K}$ and $Y_K$, it follows that

\begin{align*}
    Y_{-K} \,|\, Y_K = y_k \sim \normnot{c \, y_k \, \iota}{\tilde{\Sigma}}. 
\end{align*}

Let $h(\mu) = \expe{X | X \in B_{-K}, \, X \sim \normnot{\mu}{\tilde{\Sigma}}}$ for $B_{-K} = \{ \tilde{x} \in \reals^{K-1} | \xt_j \leq \tilde{x}_j \leq b_j, \text{ for } j = 1,\ldots, K-1 \}$. Then $\expe{ \iota' Y_{-K} \,|\, \xt_j \leq Y_j \leq b_j, \text{ for } j<K, \, Y_K = y_k } = \iota' h(c y_k \iota)$. Hence,

\begin{align*}
    \dfrac{\partial}{\partial y_k} \expe{ \iota' Y_{-K} \,|\, \xt_j \leq Y_j \leq b_j, \text{ for } j<K, \, Y_K = y_k } &= \iota' \left(D_\mu h |_{\mu = c y_k \iota} \right) \iota  \cdot c \\
    &= \iota' \var{Y_{-K} \,|\, Y_{-K} \in B_{-K}, Y_K = y_k } \tilde{\Sigma}^{-1} \iota c \\
    &= \iota' \var{Y_{-K} \,|\, Y_{-K} \in B_{-K}, Y_K = y_k } \iota c_1 c \\
    &\geq 0
\end{align*}

\noindent where the second line follows from Lemma \ref{lem: definiteness of jacobian of truncated expectation}; the third line uses Lemma \ref{lem: Covariance structure assumption implies iota an eigenvector} to obtain that $\tilde{\Sigma}^{-1} \iota = \iota c_1$ for $c_1 >0$; and the inequality follows from the fact that $\var{Y_{-K} \,|\, Y_{-K} \in B_{-K}, Y_K = y_k }$ is positive semi-definite and $c_1$ and $c$ are positive by construction. Thus, for all $y_k \in [\xt_k, b_k]$,

\begin{align*}
& \expe{Y_1 + \ldots + Y_{K-1} \,|\, \xt_j \leq Y_j \leq b_j \text{ for } j < K, \, Y_K = y_k } \geq \\
&\expe{Y_1 + \ldots + Y_{K-1} \,|\, \xt_j \leq Y_j \leq b_j \text{ for } j < K, \, Y_K = \xt_k }.    
\end{align*}

By the law of iterated expectations, we have

\begin{align*}
& \expe{Y_1 + \ldots + Y_{K-1} \,|\, \xt_j \leq Y_j \leq b_j, \forall j} = \\
&\expe{ \expe{ Y_1 + \ldots + Y_{K-1} \,|\, \xt_j \leq Y_j \leq b_j \text{ for } j<K, Y_K } \,|\, \xt_j \leq Y_j \leq b_j, \forall j} \geq \\
&\expe{ \expe{ Y_1 + \ldots + Y_{K-1} \,|\, \xt_j \leq Y_j \leq b_j \text{ for } j<K, Y_K = \xt_K } \,|\, \xt_j \leq Y_j \leq b_j, \forall j} =\\
&\expe{ Y_1 + \ldots + Y_{K-1} \,|\, \xt_j \leq Y_j \leq b_j \text{ for } j<K, Y_K = \xt_K }
\end{align*}

\noindent as we wished to show.

\end{proof}

\paragraph{Proof of Proposition \ref{prop: relative magnitudes under lienar pre-trends}}

From Proposition \ref{prop: bias when parallel trends is false - general}, Part \ref{subprop: betahat biased when pretrends false}, the first result is equivalent to showing that $$\Sigma_{12} \Sigma_{22}^{-1} \, \expe{\betahatpre - \betapre \,|\, \betahatpre \in B} > 0.$$ 
By Lemma \ref{lem: Covariance structure assumption implies iota an eigenvector}, Assumption \ref{Assumption: covariance matrix structure} implies that $\Sigma_{12} \Sigma_{22}^{-1} = c_1 \iota'$ for $c_1>0$, so it suffices to show that $\iota' \expe{\betahatpre - \betapre \,|\, \betahatpre \in B} $ is elementwise greater than zero. Note that by assumption $(\betahatpre - \betapre) \sim \normnot{0}{\Sigma_{22}}$. Additionally, $\betahatpre \in  B_{NS} = \{ \hat{\beta}_{pre} : | \hat{\beta}_{pre,j} | / \sqrt{\Sigma_{jj}} \leq c_{\alpha} \text{ for all } j \}$ iff $(\betahatpre - \betapre) \in \tilde{B}_{NS} = \{ \beta :  a_j \leq \beta_j \leq b_j \}$ for $a_j = -c_\alpha \sqrt{\Sigma_{jj}} - \beta_{pre,j}$ and $b_j = c_\alpha \sqrt{\Sigma_{jj}} - \beta_{pre,j}$. Since $\beta_{pre,j} < 0$ for all $j$, we have that $-b_j < a_j < b_j$ for all $j$. The first result then follows immediately from Lemma \ref{lem: mean of normal shifted upwards positive}.

From Proposition \ref{prop: bias when parallel trends is false - general}, Part \ref{subprop: betatilde biased when pretrends false}, the second result is equivalent to $\Sigma_{12}\Sigma_{22}^{-1} \betapre < 0$. But by assumption $\betapre < 0$, and we've shown that $\Sigma_{12}\Sigma_{22}^{-1} = c_1 \iota' > 0$, from which the result follows. $\Box$ 

\paragraph{Proof of Proposition \ref{prop: TN distribution}}
\begin{proof}
Follows immediately from Lemma 5.1 and Theorem 5.2 in \citet{lee_exact_2016}.
\end{proof}

\paragraph{Proof of Proposition \ref{prop: proper coverage of TN estimators}}
\begin{proof}

The result follows immediately from Theorem 5 of \citet{andrews_identification_2017}, provided we verify that the distribution of $\eta' \betahat \,|\, Z, A \betahat \leq b $ is continuous for almost every $Z$. 

Note that by Proposition \ref{prop: TN distribution}, $\eta' \betahat \,|\, Z = z, A \betahat \leq b $ is truncated normal with truncation points $V^-(z)$ and $V^+(z)$, and hence, continuous if $V^-(z) < V^+(z)$. Since conditional on $A \betahat \leq b$ and $Z=z$, $V^-(z) \leq \eta' \betahat \leq V^+(z)$, we have $V^-(z) = V^+(z)$ only if $V^-(z) = \eta' \betahat$.

It thus suffices to show that $\prob{\eta' \betahat = V^-(Z) \,|\, A\betahat\leq b } = 0$. Note though that 

\begin{align*}
    \prob{\eta' \betahat = V^-(Z) } &= \expe{ \prob{\eta' \betahat = V^-(z) \,|\, Z =z} } \\
    &= 0
\end{align*}

\noindent where for any fixed value $z$, $\prob{\eta' \betahat = V^-(z) \,|\, Z =z } = 0$ since $\eta' \betahat$ and $Z$ are independent by construction (see \citet{lee_exact_2016} for details) and the distribution of $\eta' \betahat$ is continuous since $\betahat$ is normally distributed, $\Sigma$ is full rank, and $\eta \neq 0$. It follows that

\begin{align*}
     0 &= \prob{\eta' \betahat = V^-(Z) } \\
     &=  \prob{\eta' \betahat = V^-(Z) \,|\, A\betahat\leq b } \prob{A\betahat\leq b} + \prob{\eta' \betahat = V^-(Z) \,|\, A\betahat\not\leq b } \prob{A\betahat\not\leq b} \\
     &\geq \prob{\eta' \betahat = V^-(Z) \,|\, A\betahat\leq b } \prob{A\betahat\leq b}.
\end{align*}

Since $\prob{A\betahat\leq b} >0$ by assumption, it follows that $\prob{\eta' \betahat = V^-(Z) \,|\, A\betahat\leq b }=0$, as needed. 
\end{proof}

\begin{lem}\label{lem: relationship between gamma and beta}
Consider the two regression specifications:

\begin{align} \label{eq: dynamic regression specification - in lemma}
y_{it} = \alpha_t + \sum_{s \neq 0} \beta_s \times Treatment_i \times 1[t =s]  + \beta_{main} \, Treatment_i + \epsilon_{it}
\end{align}

\noindent and

\begin{align}
y_{it} = \alpha_t + \sum_{s > 0} \gamma_s \times Treatment_i \times 1[t =s]  + \sum_{p=0}^{P} \delta_p \times Treatment_i \times t^p    + \epsilon_{it}    \label{eq: regression specification with polynomial pretrend - in lemma}
\end{align}

\noindent where the latter allows for a $P$th order polynomial trend. Assume $P \leq K$ (so that the pre-trend coefficients are identified). Then for any $m>0$, the population OLS coefficient $\gamma_m$ from specification (\ref{eq: regression specification with polynomial pretrend - in lemma}) can be written as $\beta_m - \eta_{\gamma,m}' \betapre $ for $\beta_m$ and $\betapre$ from regression (\ref{eq: dynamic regression specification - in lemma}). Additionally, let $X_{p} = (0^p, \ldots, -K^p)$ and let $X$ be the $(K+1) \times (P+1)$ matrix with columns $X_0, \ldots X_P$. Let $M_{-1}$ be the selection matrix that selects all but the first column. Then

\begin{align}
\gamma_m = \beta_m - (m^0, \ldots, m^P)' (X'X)^{-1} X' M_{-1} \betapre \label{eq: relation of gamma_s to beta}
\end{align}
provided that the data-generating process puts equal weight on all periods in the sense that a randomly-drawn observation $(i,t)$ is equally likely to come from any period, and that $\expe{Treatment_i \,|\, t(i)} = \expe{Treatment_i}$.

\end{lem}
\begin{proof}

Let $\expestar{Y | Z }$ denote the best linear predictor of $Y$ given $Z$. Note that if $\expestar{Y_1 | Z} = Z \delta_1$ and  $\expestar{Y_2 | Z} = Z \delta_2$, then $\expestar{Y_1 + Y_2 | Z } = Z (\delta_1 + \delta_2)$.\footnote{This is because 
\begin{align*}
\delta &= \expe{X'X}^{-1} \expe{X'Y}\\
&= \expe{X'X}^{-1} \expe{X'(Y_1 + Y_2)} \\
&= \underbrace{\expe{X'X}^{-1} \expe{X'Y_1}}_{= \delta_{1}} + \underbrace{\expe{X'X}^{-1} \expe{X'Y_2}}_{= \delta_{2}}. 
\end{align*}
See also FN 20 in \citet{borusyak_revisiting_2016}.} To prove the Lemma, we will divide $y_{it}$ into three components. We will show that the regression of the first component on the RHS variables of (\ref{eq: regression specification with polynomial pretrend - in lemma}) has coefficient $0$ on $1[t=m] \times Treatment_i$; the second component has coefficient $\beta_m$ on $1[t=m] \times Treatment_i$; and the third component has a coefficient on $1[t=m] \times Treatment_i$ that is a linear combination of $\betapre$, and is equal to $- (m^0,..., m^P)' (X'X)^{-1} X' M_{-1} \betapre$ if the DGP puts equal weight on all periods.

Note that from (\ref{eq: dynamic regression specification - in lemma}), we can write:

\begin{align*} 
y_{it} &= \underbrace{ \alpha_t + \beta_{main} \, Treatment_i + \epsilon_{it}}_{=:A} +\\ 
&\underbrace{\sum_{s > 0} \beta_s \times Treatment_i \times 1[t =s]}_{=:B}  + \\
&\underbrace{ \sum_{s < 0} \beta_s \times Treatment_i \times 1[t =s]}_{=:C}
\end{align*}

\noindent where by construction, $\epsilon_{it}$ is orthogonal to all of the RHS variables of (\ref{eq: dynamic regression specification - in lemma}). Note also that all of the variables on the RHS of (\ref{eq: regression specification with polynomial pretrend - in lemma}) can be written as linear combinations of the RHS variables of (\ref{eq: dynamic regression specification - in lemma}), since all of the variables except for the polynomial terms are directly included in (\ref{eq: dynamic regression specification - in lemma}), and $Treatment_i \times t^p = \sum_s s^p \times Treatment_i \times 1[s=t]$.  Hence, $\epsilon_{it}$ is orthogonal to all the RHS variables of (\ref{eq: regression specification with polynomial pretrend - in lemma}) as well.

It is then clear that a regression of A on the RHS variables of (\ref{eq: regression specification with polynomial pretrend - in lemma}) will load entirely on the year effects and the main effect, and thus will have coefficient 0 on $1[t=m] \times Treatment_i$. Likewise, a regression of $B$ on the RHS variables of (\ref{eq: regression specification with polynomial pretrend - in lemma}) will fit perfectly and have coefficient $\beta_m$ on $1[t=m] \times Treatment_i$.

Now, consider the regression of C on the RHS variables of (\ref{eq: regression specification with polynomial pretrend - in lemma}). Note that $C$ is a linear function of $\betapre$, and therefore all of the regression coefficients from this regression will be a linear combination of $\betapre$, from which the first result follows. For the remainder of the proof, we assume that the DGP puts equal weight on each period as defined in the statement of the Lemma. By the Frisch-Waugh-Lovell theorem, we can first residualize the outcome and all of the remaining covariates against the year dummies without changing the coefficients on the other variables. The regression then becomes:

\begin{align*}
 \sum_{s < 0} \beta_s \times (Treatment_i - \overline{Treatment}) \times 1[t =s] &= 
 \sum_{s > 0} \tilde{\gamma}_s \times (Treatment_i - \overline{Treatment}) \times 1[t =s] +\\
 &\sum_{p =0}^P  \tilde{\delta}_p \times (Treatment_i - \overline{Treatment}) \times t^p  + \tilde{\epsilon}_{it}
\end{align*}

\noindent where $\overline{Treatment}$ denotes the expected treatment probability. Additionally, note that the LHS is always 0 for observations with $t>0$, and that $\tilde{\gamma}_s$ gets positive weight only when $t=s>0$.  It follows that $\tilde{\delta}$ will be pinned down entirely by the observations with $t\leq 0$, and then we will have 

\begin{align}
 \tilde{\gamma}_s = -(s^0,...,s^P)' (\tilde{\delta}_0,...,\tilde{\delta}_P) \label{eq: tildegamma in terms of tilde delta} 
\end{align}
\noindent so that the regression fits perfectly for observations with $t>0$. To determine the values of $\tilde{\delta}$, we can therefore focus on the regression

\begin{align}
 \beta_t (Treatment_i - \overline{Treatment}) =& \sum_{p =0}^P  (Treatment_i - \overline{Treatment}) \times t^p \times \tilde{\delta}_p  + \tilde{\epsilon}_{it} \label{eq: regression of betat on deltas}
\end{align}

\noindent restricted to observations with $t\leq0$, where we define $\beta_0 = 0$.

Consider first the problem of choosing $\tilde{\delta}$ to minimize expected MSE for only the observations with $Treatment_i = 1$. Then the residuals are of the form: $(1-\overline{Treatment})^2 (\beta_t - \sum_p \tilde{\delta}_p \times t^p)^2$. Since by assumption we are equally likely to draw a treatment observation from any period, the expected mean squared error is minimized by setting $\tilde{\delta}$ equal to the coefficient from a regression of $\twovec{0}{\betapre}$ on $X$, i.e. $\tilde{\delta} = (X'X)^{-1} X' \twovec{0}{\betapre}$, or equivalently, $\tilde{\delta} = (X'X)^{-1} X' M_{-1} \betapre$. Likewise, consider minimizing expected MSE for only the observations with $Treatment_i = 0$. Then the residuals are of the form: $(-\overline{Treatment})^2 (\beta_t - \sum_p \tilde{\delta}_p \times t^p)^2$. Again, since a control observation is equally likely to come from any period, expected MSE is minimized in this group for $\tilde{\delta} = (X'X)^{-1} X' M_{-1} \betapre$. Since this value of $\tilde{\delta}$ minimizes expected MSE for both subgroups, it must minimize overall MSE. Combining this expression for $\tilde{\delta}$ with (\ref{eq: tildegamma in terms of tilde delta}) gives the desired result. 

\end{proof}

\end{document}

%% file: Preambles/header.tex
    \usepackage[T1]{fontenc}
	\usepackage[utf8]{inputenc}
	\usepackage[english]{babel}
	\usepackage{amsfonts}
	\usepackage{amssymb}
	\usepackage{enumitem}
	\usepackage{bbm}
	\usepackage{mathtools}
	\usepackage{subfig}
	\usepackage{booktabs}	
	\usepackage{setspace}
	\usepackage{graphicx}
	\usepackage{fullpage}
	
	\usepackage{soul}
	\usepackage{xspace}
	\usepackage{csquotes}
	\usepackage{enumitem}
	\usepackage[margin = 1in]{geometry}
	\usepackage{placeins}
	\usepackage{floatrow}
	\usepackage{natbib}
	\usepackage{bibunits}
	\defaultbibliography{Bibliography.bib}  
	\defaultbibliographystyle{apalike}

	\usepackage{hyperref}
	\hypersetup{
    colorlinks=true,
    linkcolor=black,
    filecolor=magenta,      
    urlcolor=blue,
    citecolor = black
}
	\usepackage{cleveref}

	\usepackage[framemethod=tikz]{mdframed}
	\usepackage{todonotes}
	\presetkeys{todonotes}{color=black!5,inline}{} 

    \usepackage{tcolorbox}
    \newtcbox{\feedback}{nobeforeafter,colframe=black,colback=white,boxrule=0.5pt,arc=2pt,
      boxsep=0pt,left=2pt,right=2pt,top=2pt,bottom=2pt,tcbox raise base}
      
    \usepackage{amsthm}

    \newtheorem{asm}{Assumption}
    
    \newtheorem{prop}{Proposition}
    \newtheorem{lem}{Lemma}
    \newtheorem{cor}{Corollary}
    
    \newtheorem{defn}{Definition}

\usepackage{array}
\newcolumntype{L}[1]{>{\raggedright\let\newline\\\arraybackslash}m{#1}}
\newcolumntype{C}[1]{>{\centering\let\newline\\\arraybackslash\hspace{0pt}}m{#1}}
\newcolumntype{R}[1]{>{\raggedleft\let\newline\\\arraybackslash\hspace{0pt}}m{#1}}

%% file: Preambles/notation.tex
	\newcommand{\Cov}{\mathop{}\!\textnormal{Cov}}

	\newcommand{\bracks}[1]{\left[#1\right]}
	\newcommand{\expe}[1]{\mathbb{E}\bracks{#1}}
	\newcommand{\expestar}[1]{\mathbb{E}^*\bracks{#1}}

	\newcommand{\var}[1]{\mathbb{V}\text{ar}\bracks{#1}}
	\newcommand{\parens}[1]{\left(#1\right)}

    \newcommand{\prob}[1]{\mathbb{P}\parens{#1}}

	\newcommand{\reals}{\mathbb{R}}

	\newcommand{\betahat}{\ensuremath{\hat{\beta}}}
	\newcommand{\betatilde}{\ensuremath{\tilde{\beta}}}
	
\newcommand{\normnot}[2]{\mathcal{N}\parens{#1,\,#2}}

\newcommand{\twovec}[2]{\parens{\begin{array}{c} #1 \\ #2 \end{array}}}

\newcommand{\twobytwomat}[4]{\parens{\begin{array}{cc} #1 & #2 \\ #3 & #4 \end{array}}}

\newcommand{\betahatpost}{\betahat_{post}}
\newcommand{\betahatpostTN}{\betahat_{post}^{TN}}
\newcommand{\betahatpre}{\betahat_{pre}}
\newcommand{\betatildepost}{\betatilde_{post}}

\newcommand{\betapost}{\beta_{post}}
\newcommand{\gammapost}{\gamma_{post}}
\newcommand{\gammahatpost}{\hat{\gamma}_{post}}
\newcommand{\gammahatpostTN}{\hat{\gamma}_{post}^{TN}}
\newcommand{\CITNbeta}{CI^{TN}_{\beta}}
\newcommand{\CITNgamma}{CI^{TN}_{\gamma}}
\newcommand{\CITrad}{CI^{Trad}}
\newcommand{\betapre}{\beta_{pre}}
\newcommand{\ybar}{\bar{y}}
\newcommand{\xt}{x^t}

%% file: Tables/BetaPostandTilde_SummaryStats_WithoutPretrends.tex
    0 & -0.000 & 0.127 & 0.127 & 0.050 & - & - & - & - \\ 
     1 & -0.000 & 0.127 & 0.123 & 0.043 & -0.000 & 0.110 & 0.110 & 0.050 \\ 
     2 & -0.000 & 0.127 & 0.120 & 0.039 & -0.000 & 0.103 & 0.103 & 0.050 \\ 
     3 & -0.000 & 0.127 & 0.118 & 0.035 & -0.000 & 0.100 & 0.100 & 0.050 \\ 
     4 & -0.000 & 0.127 & 0.116 & 0.032 & -0.000 & 0.098 & 0.098 & 0.050 \\ 
     5 & -0.000 & 0.127 & 0.115 & 0.030 & -0.000 & 0.097 & 0.097 & 0.050 \\ 
     6 & -0.000 & 0.127 & 0.113 & 0.028 & -0.000 & 0.096 & 0.096 & 0.050 \\ 
     7 & 0.000 & 0.127 & 0.112 & 0.027 & -0.000 & 0.095 & 0.095 & 0.050 \\ 
     8 & 0.000 & 0.127 & 0.111 & 0.026 & -0.000 & 0.094 & 0.094 & 0.050 \\

%% file: Tables/BetaPostandTilde_SummaryStats_LinearPretrend.tex
    0 & 1.000 & 0.065 & 0.127 & 0.127 & 0.050 & 0.081 & - & - & - & - & - \\ 
     1 & 0.920 & 0.073 & 0.127 & 0.122 & 0.043 & 0.081 & 0.097 & 0.110 & 0.110 & 0.060 & 0.144 \\ 
     2 & 0.780 & 0.088 & 0.127 & 0.118 & 0.039 & 0.090 & 0.130 & 0.103 & 0.103 & 0.096 & 0.242 \\ 
     3 & 0.578 & 0.109 & 0.127 & 0.114 & 0.042 & 0.112 & 0.162 & 0.100 & 0.100 & 0.164 & 0.367 \\ 
     4 & 0.352 & 0.136 & 0.127 & 0.110 & 0.057 & 0.154 & 0.195 & 0.098 & 0.098 & 0.262 & 0.510 \\ 
     5 & 0.168 & 0.167 & 0.127 & 0.107 & 0.089 & 0.224 & 0.227 & 0.097 & 0.097 & 0.388 & 0.652 \\ 
     6 & 0.059 & 0.201 & 0.127 & 0.105 & 0.143 & 0.327 & 0.260 & 0.096 & 0.095 & 0.528 & 0.774 \\ 
     7 & 0.015 & 0.238 & 0.127 & 0.102 & 0.230 & 0.454 & 0.292 & 0.095 & 0.094 & 0.670 & 0.870 \\ 
     8 & 0.003 & 0.278 & 0.127 & 0.100 & 0.367 & 0.608 & 0.327 & 0.094 & 0.093 & 0.799 & 0.940 \\

%% file: Tables/TNversusTraditional_MedianBias_Table.tex
    1 & 0.000 & 0.000 & -0.000 & 0.073 & 0.065 & -0.001 \\ 
     2 & 0.000 & 0.000 & 0.000 & 0.087 & 0.065 & 0.000 \\ 
     3 & 0.000 & 0.000 & 0.000 & 0.108 & 0.065 & 0.000 \\ 
     4 & 0.000 & 0.000 & -0.000 & 0.135 & 0.065 & -0.000 \\ 
     5 & 0.000 & 0.000 & -0.000 & 0.166 & 0.064 & -0.002 \\ 
     6 & 0.000 & 0.000 & 0.000 & 0.201 & 0.062 & -0.002 \\ 
     7 & 0.000 & 0.000 & 0.000 & 0.237 & 0.063 & 0.002 \\ 
     8 & 0.000 & 0.000 & 0.000 & 0.276 & 0.072 & 0.007 \\

%% file: Tables/TN_versus_Standard_CIs.tex
    1 & 0.043 & 0.050 & 0.496 & 0.517 & 0.043 & 0.050 & 0.496 & 0.521 \\ 
     2 & 0.039 & 0.050 & 0.496 & 0.549 & 0.039 & 0.050 & 0.496 & 0.580 \\ 
     3 & 0.035 & 0.050 & 0.496 & 0.582 & 0.042 & 0.051 & 0.496 & 0.678 \\ 
     4 & 0.032 & 0.050 & 0.496 & 0.613 & 0.057 & 0.051 & 0.496 & 0.820 \\ 
     5 & 0.030 & 0.050 & 0.496 & 0.642 & 0.089 & 0.052 & 0.496 & 1.018 \\ 
     6 & 0.028 & 0.050 & 0.496 & 0.670 & 0.143 & 0.055 & 0.496 & 1.262 \\ 
     7 & 0.027 & 0.050 & 0.496 & 0.696 & 0.230 & 0.054 & 0.496 & 1.459 \\ 
     8 & 0.026 & 0.050 & 0.496 & 0.721 & 0.367 & 0.057 & 0.496 & 1.974 \\

%% file: Tables/TNAdjusted_versus_Standard_CIs.tex
    1 & 0.043 & 0.050 & 0.496 & 1.062 & 0.081 & 0.050 & 0.496 & 1.081 \\ 
     2 & 0.039 & 0.050 & 0.496 & 0.773 & 0.090 & 0.050 & 0.496 & 0.839 \\ 
     3 & 0.035 & 0.050 & 0.496 & 0.651 & 0.112 & 0.050 & 0.496 & 0.784 \\ 
     4 & 0.032 & 0.050 & 0.496 & 0.581 & 0.154 & 0.050 & 0.496 & 0.796 \\ 
     5 & 0.030 & 0.050 & 0.496 & 0.536 & 0.224 & 0.050 & 0.496 & 0.843 \\ 
     6 & 0.028 & 0.050 & 0.496 & 0.504 & 0.327 & 0.051 & 0.496 & 0.906 \\ 
     7 & 0.027 & 0.050 & 0.496 & 0.481 & 0.454 & 0.049 & 0.496 & 0.975 \\ 
     8 & 0.026 & 0.050 & 0.496 & 0.463 & 0.608 & 0.048 & 0.496 & 1.053 \\

%% file: ms.bbl
\begin{thebibliography}{}

\bibitem[Abadie, 2005]{abadie_semiparametric_2005}
Abadie, A. (2005).
\newblock Semiparametric {{Difference}}-in-{{Differences Estimators}}.
\newblock {\em The Review of Economic Studies}, 72(1):1--19.

\bibitem[Andrews and Kasy, 2017]{andrews_identification_2017}
Andrews, I. and Kasy, M. (2017).
\newblock Identification of and {{Correction}} for {{Publication Bias}}.
\newblock Working Paper 23298, National Bureau of Economic Research.

\bibitem[Athey and Imbens, 2006]{athey_identification_2006}
Athey, S. and Imbens, G.~W. (2006).
\newblock Identification and {{Inference}} in {{Nonlinear
  Difference}}-in-{{Differences Models}}.
\newblock {\em Econometrica}, 74(2):431--497.

\bibitem[Bertrand et~al., 2004]{bertrand_how_2004}
Bertrand, M., Duflo, E., and Mullainathan, S. (2004).
\newblock How {{Much Should We Trust Differences}}-{{In}}-{{Differences
  Estimates}}?
\newblock {\em The Quarterly Journal of Economics}, 119(1):249--275.

\bibitem[Borusyak and Jaravel, 2016]{borusyak_revisiting_2016}
Borusyak, K. and Jaravel, X. (2016).
\newblock Revisiting {{Event Study Designs}}.
\newblock SSRN Scholarly Paper ID 2826228, Social Science Research Network,
  Rochester, NY.

\bibitem[Brodeur et~al., 2016]{brodeur_star_2016}
Brodeur, A., L{\'e}, M., Sangnier, M., and Zylberberg, Y. (2016).
\newblock Star {{Wars}}: {{The Empirics Strike Back}}.
\newblock {\em American Economic Journal: Applied Economics}, 8(1):1--32.

\bibitem[Camerer et~al., 2016]{Camereraaf0918}
Camerer, C.~F., Dreber, A., Forsell, E., Ho, T.-H., Huber, J., Johannesson, M.,
  Kirchler, M., Almenberg, J., Altmejd, A., Chan, T., Heikensten, E.,
  Holzmeister, F., Imai, T., Isaksson, S., Nave, G., Pfeiffer, T., Razen, M.,
  and Wu, H. (2016).
\newblock Evaluating replicability of laboratory experiments in economics.
\newblock {\em Science}.

\bibitem[Cartinhour, 1990]{cartinhour_one-dimensional_1990}
Cartinhour, J. (1990).
\newblock One-dimensional marginal density functions of a truncated
  multivariate normal density function.
\newblock {\em Communications in Statistics-theory and Methods - COMMUN
  STATIST-THEOR METHOD}, 19:197--203.

\bibitem[Christensen and Miguel, 2016]{christensen_transparency_2016}
Christensen, G.~S. and Miguel, E. (2016).
\newblock Transparency, {{Reproducibility}}, and the {{Credibility}} of
  {{Economics Research}}.
\newblock Working Paper 22989, National Bureau of Economic Research.

\bibitem[Donald and Lang, 2007]{donald_inference_2007}
Donald, S.~G. and Lang, K. (2007).
\newblock Inference with {{Difference}}-in-{{Differences}} and {{Other Panel
  Data}}.
\newblock {\em The Review of Economics and Statistics}, 89(2):221--233.

\bibitem[Giles and Giles, 1993]{giles_pre-test_1993}
Giles, J.~A. and Giles, D. E.~A. (1993).
\newblock Pre-{{Test Estimation}} and {{Testing}} in {{Econometrics}}: {{Recent
  Developments}}.
\newblock {\em Journal of Economic Surveys}, 7(2):145--197.

\bibitem[Lee et~al., 2016]{lee_exact_2016}
Lee, J.~D., Sun, D.~L., Sun, Y., and Taylor, J.~E. (2016).
\newblock Exact post-selection inference, with application to the lasso.
\newblock {\em The Annals of Statistics}, 44(3):907--927.

\bibitem[Leeb and P{\"o}tscher, 2005]{leeb_model_2005}
Leeb, H. and P{\"o}tscher, B.~M. (2005).
\newblock Model {{Selection}} and {{Inference}}: {{Facts}} and {{Fiction}}.
\newblock {\em Econometric Theory}, 21(1):21--59.

\bibitem[Moulton, 1990]{moulton_illustration_1990}
Moulton, B.~R. (1990).
\newblock An {{Illustration}} of a {{Pitfall}} in {{Estimating}} the
  {{Effects}} of {{Aggregate Variables}} on {{Micro Unit}}.
\newblock {\em The Review of Economics and Statistics}, 72(2):334--338.

\bibitem[Nyhan, 2015]{nyhan2015increasing}
Nyhan, B. (2015).
\newblock Increasing the credibility of political science research: {{A}}
  proposal for journal reforms.
\newblock {\em PS: Political Science \& Politics}, 48(S1):78--83.

\bibitem[{Open Science Collaboration}, 2015]{open2015estimating}
{Open Science Collaboration} (2015).
\newblock Estimating the reproducibility of psychological science.
\newblock {\em Science}, 349(6251):aac4716.

\bibitem[Petersen, 2009]{petersen_estimating_2009}
Petersen, M.~A. (2009).
\newblock Estimating {{Standard Errors}} in {{Finance Panel Data Sets}}:
  {{Comparing Approaches}}.
\newblock {\em The Review of Financial Studies}, 22(1):435--480.

\bibitem[Pfanzagl, 1994]{pfanzagl_parametric_1994}
Pfanzagl, J. (1994).
\newblock {\em Parametric {{Statistical Theory}}}.
\newblock {W. de Gruyter}.
\newblock Google-Books-ID: 1S20QgAACAAJ.

\bibitem[Rothstein et~al., 2005]{rothstein_publication_2005}
Rothstein, H.~R., Sutton, A.~J., and Borenstein, M. (2005).
\newblock Publication {{Bias}} in {{Meta}}-{{Analysis}}.
\newblock In Co-Chair, H. R.~R., Co-Author, A. J.~S., and PI, M. B. D. A.~L.,
  editors, {\em Publication {{Bias}} in {{Meta}}-{{Analysis}}}, pages 1--7.
  {John Wiley \& Sons, Ltd}.

\bibitem[Saumard and Wellner, 2014]{saumard_log-concavity_2014}
Saumard, A. and Wellner, J.~A. (2014).
\newblock Log-concavity and strong log-concavity: A review.
\newblock {\em arXiv:1404.5886 [math, stat]}.

\end{thebibliography}
